\documentclass[twoside,leqno,twocolumn]{article}  
\usepackage{ltexpprt} 
\usepackage[algo2e,ruled, noend, noline]{algorithm2e}
\usepackage{amsmath,amsfonts,amssymb}
\usepackage[small,compact]{titlesec}
\usepackage[small,it]{caption}

\usepackage{xspace}

\usepackage{paralist}

\usepackage{graphicx}

\usepackage{epstopdf}

\usepackage{boxedminipage}

\usepackage{subfig}

\usepackage[colorlinks,urlcolor=blue,citecolor=blue,linkcolor=blue]{hyperref}

\usepackage{subfig}
\usepackage{tabularx}
\usepackage{array}
\usepackage{subfig}
\usepackage{tabularx}
\usepackage{array}

\usepackage{booktabs}
\usepackage{textcase}

\usepackage{thmtools}
\usepackage{thm-restate}

\def\doctype{1}

\def\supp{0}
\newcommand{\elsewhere}{\ifnum\supp=1{supplementary file attached with the submission}\else{appendix}\fi}

\def\tsubmission{1}

\ifnum\doctype=\tsubmission
\newcommand{\full}[1]{}
\newcommand{\submit}[1]{#1}
\else
\newcommand{\full}[1]{#1}
\newcommand{\submit}[1]{}
\fi

\newcommand{\set}[1]{\{#1\}}
\newcommand{\E}{\mathbf{E}}

\newcommand{\Sec}[1]{\hyperref[sec:#1]{\S\ref*{sec:#1}}} 
\newcommand{\Eqn}[1]{\hyperref[eqn:#1]{(\ref*{eqn:#1})}} 
\newcommand{\Fig}[1]{\hyperref[fig:#1]{Fig.\,\ref*{fig:#1}}} 
\newcommand{\Tab}[1]{\hyperref[tab:#1]{Tab.\,\ref*{tab:#1}}} 
\newcommand{\Thm}[1]{\hyperref[thm:#1]{Thm.\,\ref*{thm:#1}}} 
\newcommand{\Lem}[1]{\hyperref[lem:#1]{Lem.\,\ref*{lem:#1}}} 
\newcommand{\Prop}[1]{\hyperref[prop:#1]{Prop.~\ref*{prop:#1}}} 
\newcommand{\Cor}[1]{\hyperref[cor:#1]{Cor.~\ref*{cor:#1}}} 
\newcommand{\Def}[1]{\hyperref[def:#1]{Defn.~\ref*{def:#1}}} 
\newcommand{\Alg}[1]{\hyperref[alg:#1]{Alg.\,\ref*{alg:#1}}} 
\newcommand{\Ex}[1]{\hyperref[ex:#1]{Ex.~\ref*{ex:#1}}} 
\newcommand{\Clm}[1]{\hyperref[clm:#1]{Claim~\ref*{clm:#1}}} 
\newcommand{\Step}[1]{\hyperref[step:#1]{Step~\ref*{step:#1}}} 

\newcommand{\cW}{\mathcal{W}}
\newcommand{\cE}{{\cal E}}

\newcommand{\wedgeres}{{\it w-list}}
\newcommand{\edgeres}{{\it e-list}}

\newcommand{\mgt}{{\sc MG-Triangle}}

\newcommand{\qed}{\hfill $\Box$}

\newcommand{\update}{{\tt update}}
\newcommand{\estimate}{\mgt}
\newcommand{\wed}{W}
\newcommand{\tri}{T}
\newcommand{\trans}{\tau}

\newcommand{\DBLP}{{\tt DBLP}}
\newcommand{\enron}{{\tt Enron}}
\newcommand{\flickr}{{\tt Flickr}}

\newcommand{\ignore}[1]{}

\begin{document}


\title{\  \\[-7ex] \Large Counting Triangles in Real-World Graph Streams: Dealing with Repeated Edges and 
Time Windows \vspace{-1ex}
\thanks{This work was funded by the DARPA GRAPHS
    and DOE ASCR applied math programs.  Sandia National Laboratories is a multi-program
    laboratory managed and operated by Sandia Corporation, a wholly
    owned subsidiary of Lockheed Martin Corporation, for the
    U.S. Department of Energy's National Nuclear Security
    Administration under contract DE-AC04-94AL85000.}}
\author{\large Madhav Jha \ \ \ C. Seshadhri\  \ \ Ali Pinar \thanks{Sandia National Laboratories, Livermore, CA. Email: \{mjha, scomand, apinar\}@sandia.gov; } } 
\date{\vspace{-4ex}}

\maketitle

\begin{abstract}
Real-world graphs often manifest as a massive temporal ``stream" of edges. 
The need for real-time analysis of such large graph streams has led to progress on low memory, one-pass streaming graph algorithms. 
These algorithms were designed for simple graphs,  assuming an edge is not repeated in the stream. Real graph streams however, are almost always multigraphs i.e., they contain many duplicate edges. 
The assumption of no repeated edges requires an extra pass \emph{storing all the edges} just for deduplication, which defeats the purpose of small memory algorithms.  

 We describe an algorithm, \mgt, for estimating the triangle count of a multigraph stream of edges. 
 We show that  all previous streaming algorithms for triangle counting fail for multigraph streams, despite their impressive accuracies for simple graphs. 
 The bias created by duplicate edges is a major problem, and leads these algorithms astray.
 \mgt{} avoids these biases through careful debiasing strategies and has provable theoretical guarantees and excellent empirical performance.
 \mgt{} builds on the previously introduced wedge sampling methodology.  
 Another challenge in analyzing temporal graphs is finding the right  temporal window size. \mgt{} seamlessly handles multiple time windows, and does not require  committing to any window size(s)  a priori.  We apply \mgt{} to discover fascinating transitivity and triangle trends
in real-world graph streams.

\end{abstract}

\section{Introduction} \label{sec:intro}

Many massive graphs appear in practice as a temporal \emph{stream of edges}.
People call each other on the phone, exchange emails, or co-author a paper;  computers exchange messages;  animals come in the vicinity of each other;  companies trade with each other. Each such interaction is modeled as an edge in the graph, and has a natural timestamp.

Due to the need for real-time awareness despite the volume of such transactions, 
there is much interest in processing temporal graphs using fast, limited-memory algorithms.
Formally, think of the input as a sequence of edges $e_1, e_2, \ldots, e_m$. Some of the edges may be repeated, meaning that (say) $e_1\! =\! e_{100}\! =\! e_{125}\! =\! (u,v)$.
We are interested in small space streaming algorithms that make a \emph{single pass} over the stream $e_1, e_2, \ldots, e_m$. 
Such an algorithm maintains data structures that are many orders of magnitude smaller than the stream itself.
At every timestep $t$, these data structures are updated rapidly (possible randomly). 
The algorithm computes an accurate estimate for the property of interest on the graph seen so far.
Because of the single pass and small space, the algorithm cannot revisit edges that it has forgotten. Furthermore, it cannot always determine if the new edge, $e_t$, has appeared before.
This work focuses on triangle counting in this setting.

\textbf{Graph vs multigraph:}  Previous results assume that
the edge stream forms a simple graph, and no edge is  repeated in the stream.
This is a useful assumption for algorithmic progress; yet, often false in practice.
Real-world graph streams are multigraphs, in that same edges can occur repeatedly in the data stream.
The simple graph representation is obtained by removing duplicate edges.
For example, the classic \enron{} email dataset is really a multigraph with 1.14M edges, while the underlying simple graph has only 297K edges.
Similarly, a \DBLP{} co-authorship graph recently collected is a multigraph with 3.63M edges, but the underlying simple graph  has only 2.54M edges. 
Close to 10 million edges in 
a popular dynamic \flickr{} network dataset (see~\cite{graphrepository2013, mislove-2008-flickr}) are repeated.

The assumption of simplicity is implemented in practice with an extra pass
to remove duplicate edges. \emph{This pass requires storage of the entire simple graph, which is completely ignored
in all previous work.} Indeed, if one can store the entire simple graph, there exist much better
algorithms for triangle counting~\cite{ScWa05-2,TsKaMiFa09,SePiKo13}.
We posit that for streaming algorithms to be actually useful in practice, 
multiple edges must be dealt with small space. 
There is much work on streaming graph algorithms (see surveys~\cite{NeAhKo13,McG14}). 
Yet this algorithmic work ignores important issues  such as  repeated edges and temporal aggregation that arise when looking at a real-world graph stream, as  demonstrated in \Fig{prev-multi}.

\begin{figure*}[t]
  \centering
  \subfloat[{Real stream}]{\includegraphics[width=0.4\textwidth]{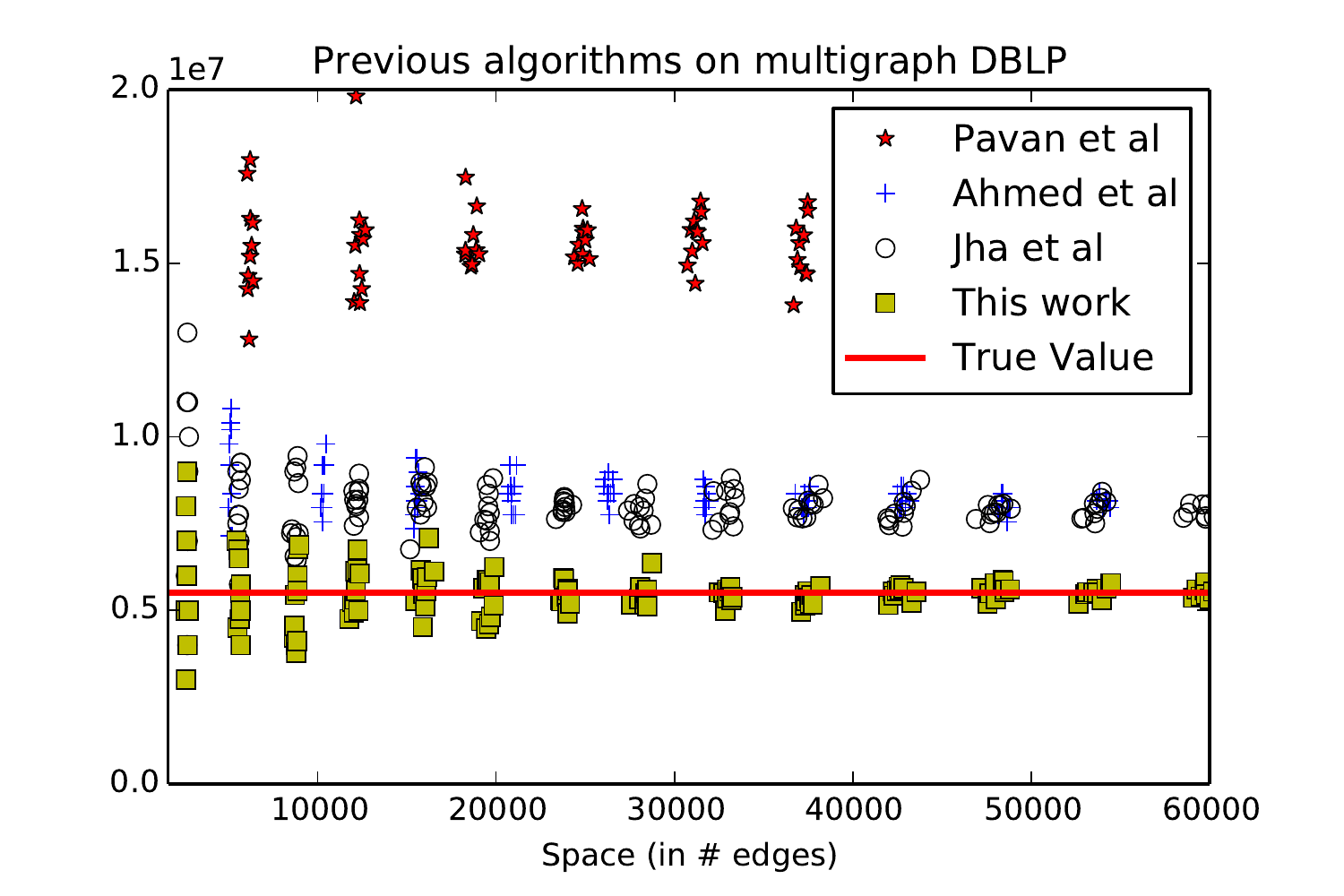}\label{fig:prev-multi}}\
  \subfloat[{Simple}]{\includegraphics[width=0.4\textwidth]{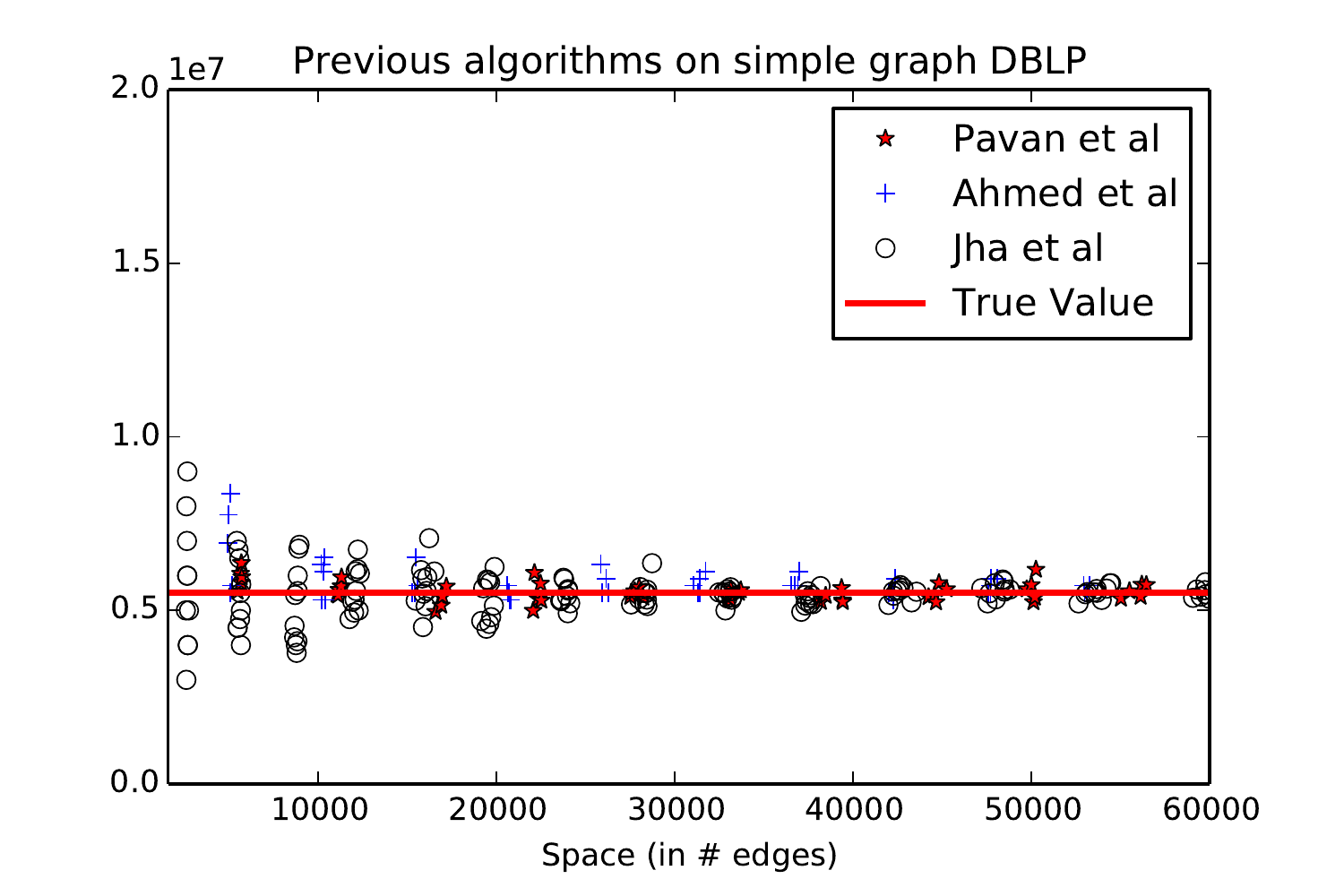}\label{fig:prev-simple}}\
   \caption{Previous works when run on multigraph \DBLP{} converge (as storage increases) to an incorrect value. On the simple graph \DBLP, they converge correctly.}
  \label{fig:prev-works}
\end{figure*}

\textbf{Aggregation over time:} Given a stream of edges, what is the actual graph? The most common answer
is to simply aggregate all edges ever seen. Again, this is a useful assumption for algorithmic progress, but
ignores the temporal aspect of the edges. Time is a complex issue and there are no clear solutions.
One may consider sliding windows in time or have some decay of edges. For simplicity, we focus on sliding
time windows (like edges seen in the past month, or past year).
Even for sliding windows,
it is not clear what the width should be. Observations can often be an artifact of the window size~\cite{Mac12}. 
Therefore, it is essential to observe multiple time windows at the same time, instead of committing to a single one. 

\Fig{dblp-window} shows how our algorithm \estimate \ can analyze  different time windows with a single run.  Our algorithm estimates the triangle counts for any time window without altering its data structures, as the time window is only used in calculating the estimate.  

\begin{figure*}[t]
  \centering
  \subfloat[{\DBLP: transitivity}]{\includegraphics[width=0.23\textwidth]{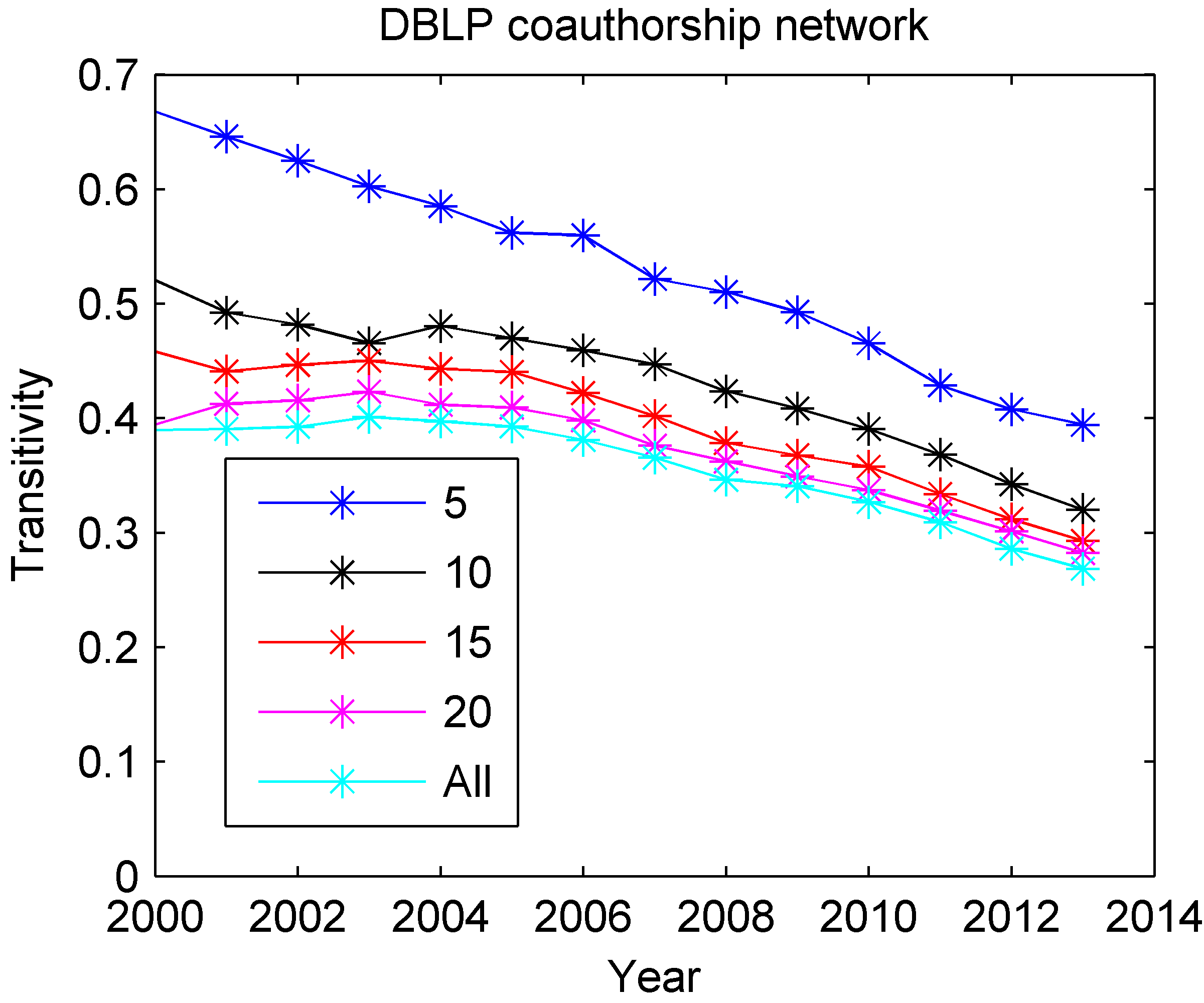}\label{fig:dblp-w-gcc} } \ 
  \subfloat[{\DBLP: triangles}]{\includegraphics[width=0.23\textwidth]{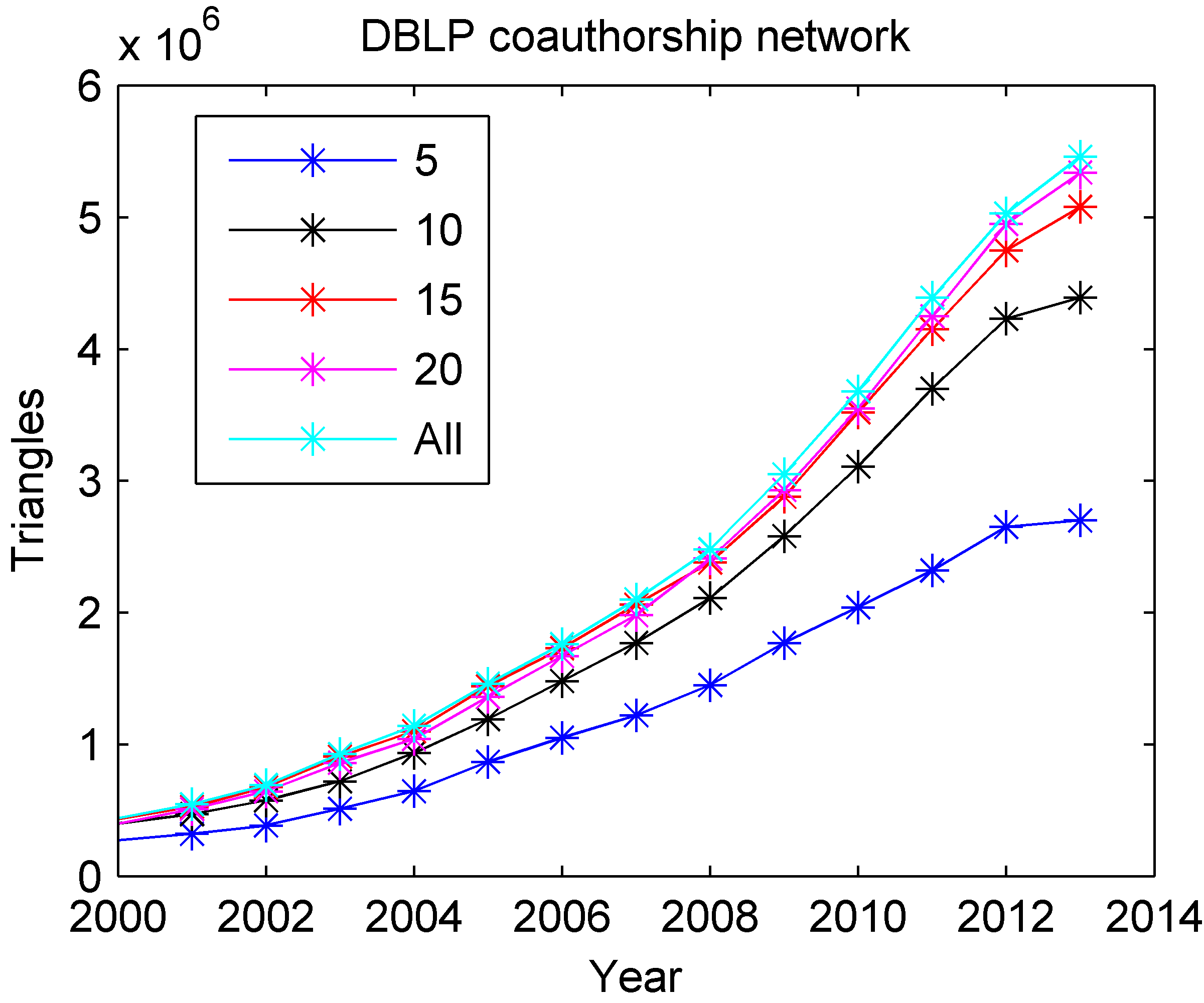}\label{fig:dblp-w-tri}}\ 
   \subfloat[{\enron: Transitivity}]{\includegraphics[width=0.23\textwidth]{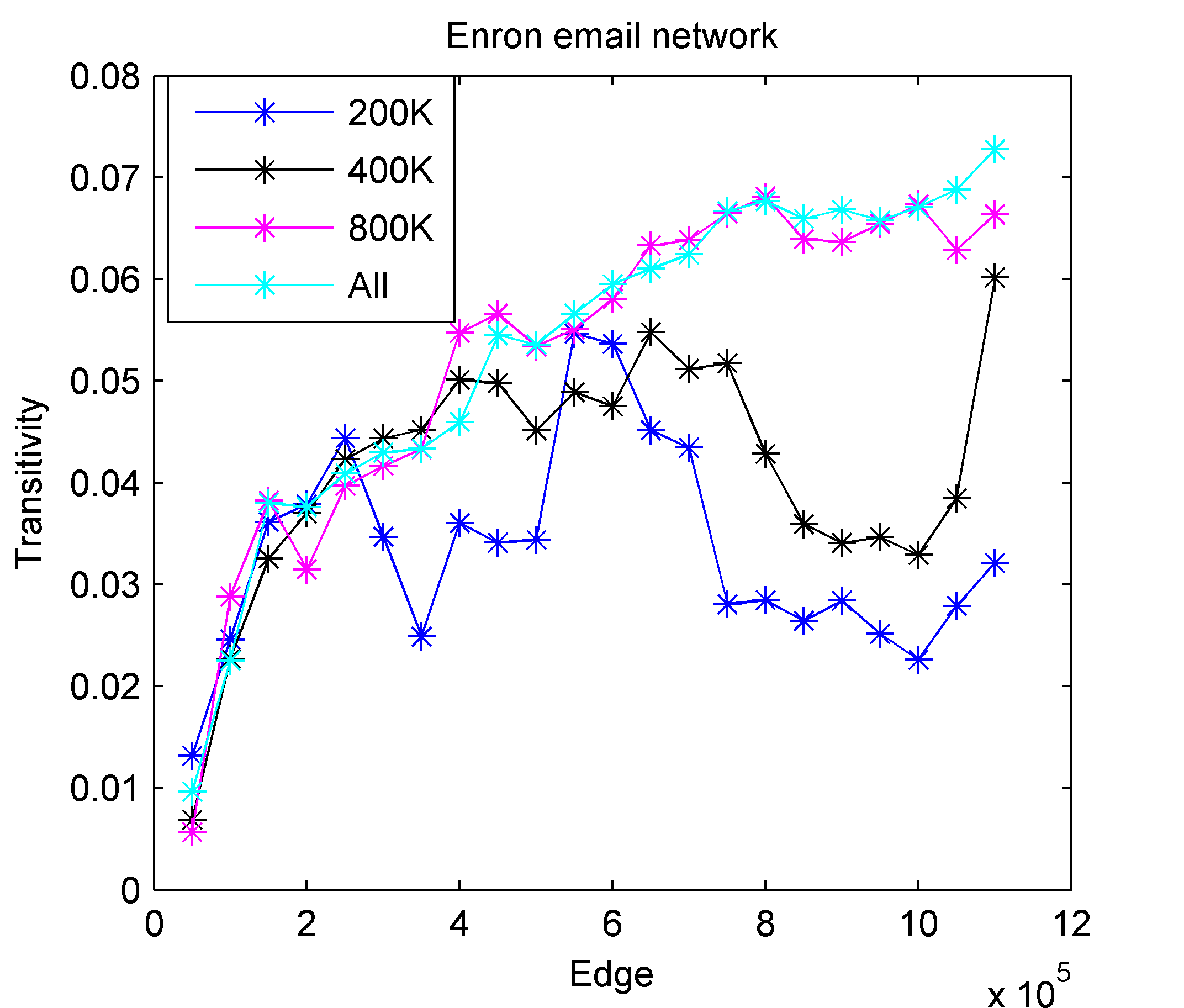}\label{fig:enron-w-gcc}} \ 
  \subfloat[{\enron: triangles}]{\includegraphics[width=0.23\textwidth]{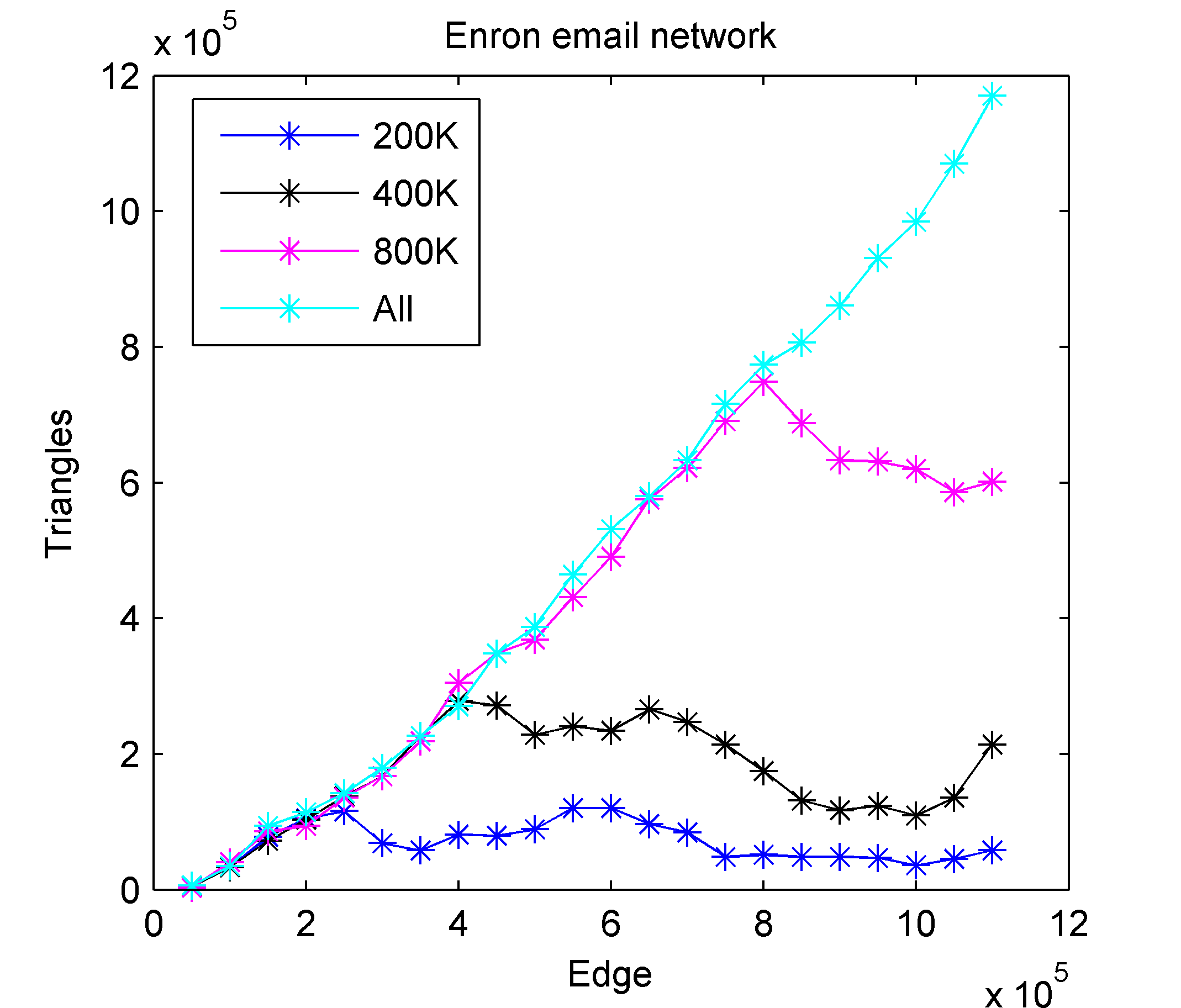}\label{fig:enron-w-tri}}

   \caption{Transitivity and triangles estimates for varying window lengths for \DBLP\  coauthorship  and \enron\  email networks: Our algorithm
   stores less than 3\% of the multigraph stream yet produces fine-grained triangle results.}
  \label{fig:dblp-window}
\end{figure*}

\subsection{Triangle counting}
The abundance of triangles has been observed in networks arising in numerous scenarios, 
such as  social sciences~\cite{Co88,Po98,Burt04,FoDeCo10}, spam detection~\cite{EcMo02}, community detection~\cite{GS12}, finding common topics on the web~\cite{BeBoCaGi08}, bioinformatics~\cite{Milo2002}, and modeling and characterizing real-world networks~\cite{SeKoPi11,DuPiKo12}. Subsequently, there has been a lot of work on triangle counting in graph streams \cite{JoGh05,BuFrLeMaSo06,AhGuMc12,KaMeSaSu12,TaPaTi13,PaTaTi+13,JhSePi13,AhDuNe+14}, and  in various other settings (see e.g., \cite{SPK14} and references therein).
The result of Ahmed et al.~\cite{AhDuNe+14} is arguably the state-of-the-art, with a storage
significantly smaller than previous algorithms. None of these results explicitly deal with multigraphs.

Formally, we are processing a multigraph
stream $e_1, e_2, \ldots, e_m$. At every time $t$, consider the underlying \emph{simple} graph $G_t$ formed by 
edges $e_{t-\Delta t},\ldots,e_t$. So take all these edges, and remove duplicates.
We wish to output the triangle count (alternately, the
transitivity) of $G_t$ for all times $t$. The window length $\Delta t$ may be defined in different ways.
It could either be in terms of number of edges (say, the past 10K edges), or in terms
of the semantics of timestamps (say, edges seen in the past month). Most importantly,
we want a single-pass small space algorithm to handle multiple windows lengths and do not
want different passes for each window length.

A reader may wonder why we only output estimates for the underlying simple graph.
Ideally, we would like to compute measures that involve the multigraph structure.
We agree that this is an interesting problem, and duplicates have their
own significance. Currently, it is standard to focus on simple graphs, and there is no consensus on how to define triadic measures
on  multigraphs. This is  an exciting avenue for future work.

\textbf{Why is this a difficult problem?} Multigraphs are a major challenge for triangle counting algorithms. Edges appears with varying frequencies, and (in our setting) we do not wish
to be biased by this. Furthermore, triangles can be formed in different ways.
Consider edges $a, b$, and $c$ that form a triangle. These edges may appear in the multigraph
stream in many different ways. For example, these edges could come as $a,a,\ldots, b,b,\ldots, c,c,\ldots$, or as $a,b,c,a,b,c,a,b,c,\ldots$.
(Observe how this is  \emph{not} an issue for simple graphs.)
These patterns create biases for existing triangle counting algorithms, which we explain in more detail
later. 

For now, it suffices to say that existing algorithms~\cite{BuFrLeMaSo06,JhSePi13,PaTaTi+13,TaPaTi13,AhDuNe+14} will give different estimates
for triangle counts of different multigraphs streams that contain the same simple graph.
This is demonstrated in \Fig{prev-multi}, where we run 
previous streaming triangle counting algorithms on the raw \DBLP{} multigraph stream.
Previous algorithms converge to an incorrect value as their storage increases.
They all perform extremely well if all duplicates were removed from the stream (\Fig{prev-simple}).
Previous work on multigraph mining explicitly states
triangle counting of streaming multigraphs as an open problem~\cite{CormodeM05}.


\subsection{Preliminaries}

The edge stream is denoted by $e_1, e_2, \ldots, e_m$. 
We focus on undirected graphs, so each edge is an unordered pair of vertex ids.
The simple graph formed by edges $e_{t'},\ldots,e_t$ is denoted by $G[t',t]$. 
A \emph{wedge} is a path of length $2$.
The set of wedges in a simple graph $G$ is denoted $\wed(G)$, and the set
of triangles by $\tri(G)$. A wedge in $\wed(G)$ is \emph{closed} if it participates
in a triangle and \emph{open} otherwise. The \emph{transitivity} is the fraction
of closed wedges, $\trans(G) = 3|\tri(G)|/|\wed(G)|$.
Our aim is to maintain the transitivity and triangle count (for all $t$) of the graph
$G[t-\Delta t,t]$, where $\Delta t$ is the desired window of aggregation.
The window is usually specified as a fixed number of edges or
a fixed interval of time (like month, year, etc.), though the algorithm works
for windows lengths that change with time. For convenience, we 
denote $G_t = G[t-\Delta t,t]$, $E_t = E(G_t)$, $W_t = \wed(G_t)$, $T_t = \tri(G_t)$, 
and $\trans_t = \trans(G_t)$.

\subsection{Our Contributions}

We design a small space streaming algorithm, \mgt, to estimate transitivity and triangle counts for multiple time windows on multigraphs. 
As mentioned earlier, the main technical contribution is in handling repeated edges without a separate storage-intensive deduplication
process.
We consider this work as a first step towards small space streaming analytics for real-world graph streams. %

\begin{asparaitem}
\item {\bf The multiedge problem:} We applied previous streaming triangle algorithms~\cite{PaTaTi+13,JhSePi13,AhDuNe+14}
on multigraph streams, and showed that they fail to give correct answers. \Fig{prev-multi} shows how all 
these algorithms converge as their storage increases to an incorrect estimate on a \DBLP{} multigraph stream.
Of course, these algorithms
were designed with the assumption of simple graph streams, and have excellent convergence properties (\Fig{prev-simple}). These results show how repeated edges
are a problem and why we need new algorithms for multigraph streams.

\item {\bf Theoretical and empirical proofs of convergence:} We give proofs
of convergence for \mgt. Our algorithm is based on wedge sampling~\cite{ScWa05-2,SePiKo13} and borrows
ideas from~\cite{JhSePi13,AhDuNe+14}. 
It is provably correct on expectation. 
We also  prove variance bounds, but \mgt{} shows much better performance in practice than such bounds would indicate. We perform detailed
experiments to prove that our algorithm gives accurate estimates
with little storage (less than 5\% of the stream in all instances). In \Fig{prev-multi}, we observe
how \mgt{} converges to the correct value storing at most 60K edges (the stream size is 3M).

\item {\bf Low storage required on real-world graphs:} Our algorithm stores less than 5\% of the stream
in all instances, and gives accurate estimates for transitivity and triangles counts.
For example, we converted a 223M edge {\tt orkut} graph~\cite{Snap} to a 500M edge multigraph, where  our algorithm produced triangles estimates within $1\%$ relative error. 
The storage required was just 1.2M edges, less than $0.5\%$ of the stream.
Our algorithm's worst performance (on a livejournal social network) only led to $0.04$ additive error in transitivity, and 8.7\% relative error in triangle count.

\item {\bf Multiple time window estimates in real-world graph streams:}  
\Fig{dblp-window} presents 
an example output of \mgt{} on a \DBLP{} coauthorship
graph stream.
\mgt{} makes a single pass and stores less than 100K edges ($<\!3\%$ of total stream).
It gives estimates for transitivity and triangles count at every year for window sizes
of 5, 10, 15, 20 years, and all of time. In other words, at year (say) 2013,
it gives triangle estimates for the simple graphs that aggregates
edges in the following intervals: 2009--2013, 2004--2013, 1999--2013, 1994--2013,
and 1938--2013.
We immediately detect specific trends for different windows, like increasing window size
decreases transitivity (even though triangle count naturally goes up). Also note the overall
decrease of transitivity over time. We also perform such analyses on an email network
and a social network, and observe differences between these graphs.
\end{asparaitem}

\section{Effects of repeated edges on triangle counting} \label{sec:effect}

We describe previous practical streaming triangle algorithms and 
explain why repeated edges is a challenge. We hope that this provides better
context for our work and explains how important the assumption of simple graphs is
for previous work. Our focus is on the neighborhood sampler of Pavan et al.~\cite{PaTaTi+13},
the wedge sampler of Jha et al.~\cite{JhSePi13}, and the sample-and-hold algorithm of Ahmed et al.~\cite{AhDuNe+14}.
To the best of our knowledge, these are the algorithms with established practical performance
and good theoretical guarantees. (We omit the algorithm of Buriol et al.~\cite{BuFrLeMaSo06}, since its 
practical performance is not good even for million edge streams~\cite{JhSePi13}.)
For the sake of exposition, we formulate and describe the algorithms in slightly different
terms from the original papers.

\textbf{Reservoir sampling vs hashing:} All algorithms sample uniform random edges
from the stream, either by reservoir sampling or sampling an edge with fixed probability, which poses 
 a problem in multigraph streams, since frequent edges
have a higher probability of being sampled.  This problem  can be mitigated by using
random hash functions. 
Suppose we wish to store each edge of the underlying simple graph from the stream with probability $\alpha$. Each edge should be equally likely to be selected, independent of its frequency.
Let $hash$ be a uniform random function into the range $(0,1)$. 
When the algorithm sees an edge $e$ in the stream, it stores the edge
if $hash(e) < \alpha$. 
Observe that the probability that an edge is selected only depends on its hash value and \emph{is independent
of its frequency}. We also stress that, for simple graph streams, hash based sampling is essentially
equivalent to any other uniform random method.

Hashing provides an easy fix for the basic sampling problem, and is actually a convenient
implementation method even for simple graphs. (We implemented all previous algorithm using hashing.)
But the real challenge is debiasing, which comes next. 

\textbf{Neighborhood sampling~\cite{PaTaTi+13}:} Let edge $f$ be a neighbor of $e$, if $e$ and $f$ intersect. 
The main idea of~\cite{PaTaTi+13} is to pick a uniform random edge $e$, and then pick a uniform random neighbor $f$
of $e$ from the subsequent edges. This provides a wedge $\{e,f\}$, which is then checked for closure
to provide a triangle. This process samples triangles non-uniformly. Pavan et al. cleverly debias by counting the number of following edges adjacent to $e$. (Equivalently, keeping track of the degree of vertices after storing $e$.) The algorithm takes a number of independent samples
to get a low-error estimate. The method is provably correct and has excellent behavior in practice.

But multigraphs affect this debiasing. Tracking (simple) degrees of a vertex $v$ is a non-trivial task,
and requires counting the number of distinct edges incident to $v$. This itself
requires a space overhead and it is not clear how to get a complete small-space extension of this approach for multigraphs.

\textbf{Sample-and-hold~\cite{AhDuNe+14} and wedge sampling~\cite{JhSePi13}:} Ahmed et al. give an
elegant algorithm for triangle counting. Simply
store every edge with some fixed (small) probability. For every edge $e$ in the stream, count the number
of triangles formed by $e$ and a wedge among the stored edges. The sum of these counts
can be used to estimate the total number of triangles. The final algorithm is simple, converges extremely rapidly, and is space efficient
(To date, it is arguably the best streaming triangle counting algorithm).
The wedge sampling algorithm of Jha et al~\cite{JhSePi13} can also be thought
of in this framework, except that it tracks a subset of the wedges created by stored edges.

Without getting into details, it suffices to say that the correctness of these algorithms
hinges on a critical fact. Every triangle (in a simple graph) stream has a unique wedge
that closes in the future. Suppose edges $\{e,f,g\}$ form a triangle, and edges
appear in order $e,\ldots,f,\ldots,g$. Then the wedge $\{e,f\}$ is closed subsequently
by edge $g$. It can be shown that both algorithms sample triangles uniformly, leading to unbiased estimates.
This is not true for multigraph streams. If they appear
in the stream as $e,f,g,e,f,g,e,f,g,\ldots$, there is no unique wedge closed in the future. (Indeed, all
wedges are closed in the future.) This is a significant problem and increasing storage does not mitigate this problem.
As demonstrated in \Fig{prev-multi}, these algorithms converge to an incorrect estimate as storage increases.

\section{Proposed algorithm}
Our algorithm \estimate{} takes as input sampling rates $\alpha, \beta \in (0,1)$
and a window $\Delta t$. The window is specified as a fixed number of edges or
a fixed interval of time (like month, year, etc.).
We describe the data structures used by \estimate.
\begin{asparaitem}
	\item Lists \edgeres, \wedgeres: These are lists consisting of random edges and wedges, respectively. 
	The sizes of these lists are controlled by $\alpha$ and $\beta$.
	\item Flags $X_w$: For each wedge $w \in$ \wedgeres, we have a boolean flag $X_w$
	supposed to denote whether it is open or closed.
\end{asparaitem}
As mentioned earlier, it is convenient to think of $hash$ as a uniform
random function into the range $(0,1)$. 
Abusing notation, we will use $hash$ to map various different
objects\footnote{This is implemented by appropriately concatenating vertex ids.}
such as edges, wedges, etc.

\begin{algorithm2e}
\caption{\estimate($\alpha,\beta,\Delta t$)}\label{algo:estimate}
\DontPrintSemicolon
\nl\ForEach{edge $e_t$ in the stream}{
	\nl Call \update($e_t$). \;
	\nl	\ForEach{wedge $w$ in  \wedgeres}{\label{step:id-begin-bias-correction}
		\nl	Let $w = \set{(u,v), (u,w)}$. \;
		\nl	\If {$e_t$ is the closing edge $(v,w)$}{
		\nl	Set $X_w$ to $1$.\;
	 	}
		\nl	\ElseIf{$e_t \in \set{(u,v), (u,w)}$}{
		\nl	Reset $X_w$ to $0$. \tcp*[r]{bias-correction} \label{step:reset}
		}
	}
	\nl Let $\cW \subseteq$ \wedgeres{} be the set of wedges that formed in time $[t-\Delta t,t]$.\;
	\nl Output $\widehat{T}_t = (\alpha^2\beta)^{-1}\sum_{w \in \cW} X_w$.\;
	\nl Output $\widehat{W}_t = (\alpha^2\beta)^{-1}|\cW|$ and $\trans_t = 3\widehat{T}_t/\widehat{W}_t$
	(if $\widehat{W}_t = 0$, set $\trans_t = 0$).
}
\end{algorithm2e}

\begin{algorithm2e}
\caption{\update($e_t$)}\label{algo:update}
\DontPrintSemicolon
\nl \If {$hash(e_t) \leq \alpha$ and $e_t \notin$ \edgeres} {
	\nl Insert $e_t$ in \edgeres. \;
	\nl \ForEach {wedge $w = (e,e_t)$ where $e \in$ \edgeres} {
		\nl \If {$hash(w) \leq \beta$ and $w \notin$ \wedgeres{}} {
			\nl Insert $w$ in \wedgeres.\;
			}
		}
	}	
\end{algorithm2e}

\subsection{High level description} 

The first step on encountering edge $e_t$
is to update the lists \edgeres{} and \wedgeres. This is done in procedure
\update. The idea is based on standard hash-based sampling. We add $e_t$
to \edgeres{} if $hash(e_t) \leq \alpha$ and $e_t$ is not already in \edgeres.
Then, we look at all the wedges that $e_t$ creates with existing
edges in \edgeres. We apply another round of hash-based sampling 
to put these wedges in \wedgeres. 

Critically, if an edge $e$ enters \edgeres, it never
leaves. If $e$ enters \edgeres{}, it does so the first time it appears in the stream.
The probability of an edge entering \edgeres{} is independent of its frequency
in the stream. This is vital to get unbiased samples of edges in the underlying simple graph $G_t$.
Similar statements hold for wedges.

\textbf{Checking for closures and debiasing:} 
We encounter edge $e_t$ and have updated \edgeres{} and \wedgeres.
For each wedge $w \in$ \wedgeres, we have a boolean variable $X_w$. If $e_t$
closes $w$ (so $w$ and $e_t$ form a triangle), we set $X_w \!=\! 1$. This is 
the standard wedge-sampling approach~\cite{ScWa05-2,SePiKo13,JhSePi13}. At this point,
the algorithm would basically be that of~\cite{JhSePi13}, implemented with hash-based
sampling. As argued earlier and shown in~\Fig{prev-works}, this algorithm does not
work.

To fix the biasing, we perform a somewhat mysterious step. 
We have wedge $w \in$ \wedgeres{} and encounter $e_t$. If $e_t$ is already
part of $w$, we simply reset $X_w$ to $0$. So even though $w$ may be closed,
we just assume it is open. This completely resolves the biasing, and we give
a formal proof in \Thm{main}.

\textbf{Outputting the estimate:} Finally, we need to output estimates,$|\widehat{T}_t|, |\widehat{W}_t|, \widehat{\trans}_t$ for $|T_t|, |W_t|, \trans_t$, respectively.
This is the only step where the time window $\Delta t$ is used. We look at all wedges
in \wedgeres{} that formed in the time $[t-\Delta t,t]$. The total number of these
wedges can be scaled to estimate $|W_t|$. The number of these wedges, where $X_w = 1$
is scaled to estimate $|T_t|$, and the appropriate ratio estimates $\trans_t$.

\subsection{Theoretical analysis} \label{sec:theory}

We prove that the \estimate{} is correct on expectation
and prove weak concentration results bounding the variance.
We also show some basic bounds on the storage of \estimate.
Throughout this section, we focus at some time $t$ and the simple graph $G_t$. 
We stress that there is no distributional assumption
on the graph or the stream. All the probabilities are over the internal randomness
of the algorithm (which is encapsulated in the random behavior of $hash$).

\begin{restatable}{lem}{problemma}
\label{lem:prob} Consider time $t$. For any edge $e \in G_t$,
the probability that $e \in$ \edgeres{} is $\alpha$.
For any wedge $w \in W_t$, the probability
that $w \in$ \wedgeres{} is  $\alpha^2\beta$.
\end{restatable}

\begin{proof} Consider edge $e$. We first argue that  $e\! \in$\! \edgeres{} iff $hash(e) \!\leq\! \alpha$
(Note that this is independent of  the frequency of $e$).
Suppose $hash(e) \!\leq \!\alpha$. At its first  occurrence,  $e$ enters \edgeres\ and remains in \edgeres.
Suppose $hash(e) \! >\!  \alpha$. At no timestep will $e$ be added to \edgeres, regardless of how many times
it appears.
From the randomness of $hash$, $hash(e) \!\leq\! \alpha$ with probability $\alpha$. Hence, $e \in$
\edgeres{} with probability $\alpha$.

%
For wedge $w\! =\! \{e,e'\}$
to be in \wedgeres, both its edges must be in \edgeres.
That means both $hash(e)$ and $hash(e')$ are at most $\alpha$.
Suppose the first occurrence of $e$
is before that of $e'$. At the first time $e'$ occurs, procedure \update{}
will add $w$  to \wedgeres\ iff $hash(w) \leq \beta$.
At any subsequent occurrence of $e$ or $e'$, the wedge $w$ is not considered
for adding to \wedgeres{} (simply because $e$ and $e'$ are already in \edgeres).
The total probability (by the randomness of $hash$) is $\alpha^2\beta$.
\qed
\end{proof}

The following hold just by linearity of expectation.
We move proofs to the \elsewhere.
%

\begin{restatable}{thm}{spacethm}
\label{thm:space}
The expected size of \edgeres{} is  $\alpha E(G[1,t])$ and the expected size of \wedgeres{} is $\alpha^2\beta W(G[1,t])$.
\end{restatable}

%
%

\begin{restatable}{thm}{wedgethm}
\label{thm:wedge} $\E[\widehat{W}_t] = |W_t|$.
\end{restatable}

%
%
Now we come to a key theorem that shows that $\widehat{T}_t$
is correct on expectation. This is where we prove that our proposed debiasing technique works.

\begin{restatable}{thm}{mainthm}
\label{thm:main} $\E[\widehat{T}_t] = |T_t|$.
\end{restatable}

\begin{proof} We extend the definition of Boolean flag $X_w$ to every wedge $w$ in $W_t$. 
Let $X_w\!=\!0$ if $w$ is not present in \wedgeres{} (at time $t$). 
Note that $\widehat{T_t} = (\alpha^2\beta)^{-1}\sum_{w \in W_t} X_w$.
For every edge $e$ in $E_t$, 
let $t_{max}(e)$ be the maximum time $s \!\leq \!t$ such that $e_s \!= \!e$. 
Fix a triangle $A \!= \!\set{a,b,c} \!\in \!T_t$ formed by edges $a, b,$ and $c$, and 
assume (by relabeling if required) that $c$ is the last edge to appear in the stream among $a, b$, and $c$. 
In other words, $t_{max}(c) > \max\set{t_{max}(a), t_{max}(b)}$. Since $\{a,b\}, \{b,c\}, \{c,a\}$ are wedges, it makes
sense to talk about $X_{\{a,b\}}$, etc.
The following is the debiasing argument, showing that exactly
one wedge in $A$ has $X_w = 1$.

\begin{lemma}
\label{lem:debias}
 $X_{\{b,c\}} = X_{\{c,a\}} = 0$. Moreover, $X_{\{a,b\}} = 1$ iff $\{a,b\}$ is in \wedgeres.
\end{lemma}

\begin{proof}
Consider the moment $s = t_{max}(c)$ when $e_s = c$. 
If wedge $\{b,c\} \not \in$ \wedgeres, then by definition, $X_{\set{b,c}}$ is $0$. 
If  $\{b,c\}\in$ \wedgeres, then by Step~\ref{step:reset} of Algorithm~\ref{algo:estimate}, the value of $X_{\{b,c\}}$ is reset to $0$. No subsequent change is made to this value. An identical argument shows the same for $X_{\{c,a\}}$. 
Finally, $X_{\{a,b\}}$ is set to 1 at this moment iff if wedge $\{a,b\}$ is in \edgeres, and once again, this value is not changed subsequently.
\qed
\end{proof}

By Lemma~\ref{lem:debias}, $\E[X_{\{b,c\}}]  = \E[X_{\{c,a\}}] = 0$, while $\E[X_{\{a,b\}}]$
is the probability that this wedge is in \wedgeres. This is exactly $\alpha^2\beta$.
Therefore, the sum of expectations of $X_w$ over all three wedges $w$ of the triangle $A = \set{a,b,c}$ is $\sum_{w \in A} \E[X_w] = \alpha^2\beta$. 
Observe this is true for any fixed triangle in $T_t$.
For any wedge $w$ that does not participate in a triangle,
$X_w$ is obviously zero. By linearity of expectation, 
$\E[\widehat{T}_t] = (\alpha^2\beta)^{-1} \E[\sum_{w \in W_t} X_w] $ $= (\alpha^2\beta)^{-1} \sum_{A \in T_t} \sum_{w \in A} \E[X_w]$.
Plugging in the value of $\E[X_w]$, this is $(\alpha^2\beta)^{-1} \cdot \alpha^2\beta |T_t| = |T_t|$.
\qed
\end{proof}

Using methods from~\cite{JhSePi13}, we can prove weak concentration bounds for
 $\widehat{T}_t$ and $\widehat{W}_t$ (by bounding their variance). 
We need to assume that  $\alpha$ and $\beta$ are large enough to ensure that enough
wedges of $W_t$ are in \wedgeres, and there are at least as many wedges
in $G_t$ as edges. The latter is needed to rule out extreme cases
like $G_t$ being a path or a matching. This assumption is reasonable
for real-world networks, as can be seen in~\Tab{run-synthetic}. Proof 
is in the \elsewhere.

\begin{restatable}{thm}{concthm}
\label{thm:conc} Fix some sufficiently small $\gamma > 0$.
Suppose that $(\alpha^2\beta)|W_t|$ (the expected number of wedges in $W_t$
that are in \wedgeres) is at least $1/\gamma^6$. Furthermore
$|W_t| \geq |E_t|$ (there are at least as many wedges in $G_t$ as edges).
Then, $\Pr[|\widehat{W}_t - |W_t|| > \gamma|W_t|] < \gamma$, 
$\Pr[|\widehat{T}_t - |T_t|| > \gamma|W_t|] < \gamma$, and
$Pr[|\widehat{\trans}_t - \trans_t| > 8\gamma] < 4\gamma$.
\end{restatable}


\section{Empirical evaluation of \estimate}

We implemented our algorithm in C++ and ran it on a MacBook Air laptop with 1.7 GHz Intel Core i7 processor and 8 GB 1600 MHz DDR3 RAM. We applied \estimate{} on a variety of real-world datasets.  Refer to \Tab{run-synthetic} for details about these datasets.\\
{\bf \DBLP:}   This is a co-authorship network for papers on the DBLP website.  From  the raw data at DBLP~\cite{DBLP} we  extracted 786,719 papers by ignoring papers with (i) a single author, (ii) more than 100 authors, and (iii) missing ``year'' metadata. For each paper we put an edge corresponding to every distinct pair of co-authors resulting in a total of 3,630,374 (multi)edges. \\
{ \bf  \enron:} This network is derived from emails between Enron employees between 1999 and 2003~\cite{graphrepository2013}. Nodes correspond to employees while edges represent  their email correspondence. Multiple emails between the same pair of individuals result in a multigraph. \\
 {\bf \flickr:} This dataset consists of friendship connections of users of Flickr, obtained from \cite{graphrepository2013}. Originally, the data was collected in \cite{mislove-2008-flickr}. (Results on \flickr{} given in the \elsewhere.) \\
{\bf SNAP:}  We extended our data set to include networks from SNAP \cite{Snap}. We synthetically replicate edges of these datasets to get a multigraph.
	
\begin{table*}[t]	
 \caption{A run of our algorithm on a variety of real-world and synthetic graphs with $\alpha = 0.01$ and $\beta$ set such that size of \wedgeres{} is at most 50K. The third column gives the number of edges in the multigraph while the fourth column (Space) gives the space (in terms of number of edges) used by the algorithm. The first three datasets are raw real-world datasets whereas the remaining datasets were synthetically made multigraphs starting with graphs from \cite{Snap}.}
\label{tab:run-synthetic}
\centering
\scriptsize
  \begin{tabular}{|l|c|c|cc|c|cc|cc|}
  \hline
\multicolumn{1}{|c|}{Dataset } & \multicolumn{1}{c|}{$n$} & \multicolumn{1}{c|}{Wedges} & \multicolumn{2}{c|}{Edges}  &  \multicolumn{1}{c|}{ Space} & \multicolumn{2}{c|}{Transitivity} & \multicolumn{2}{c|}{Triangles} \\  
 &  & (simple) & simple & multi &  &   exact &  estimate &  exact & Rel. error \\ \hline
\DBLP 			&  755K	&  61M 		&  2.54M		&  3.63M 		 & 31K 	& 0.269 	& 0.282	&5.50M	& 3.09\%\\
\enron 			& 86K	&  49M 		& 297K 		&  1.15M		&  8K 	& 0.069 	& 0.071  	& 1.18M 	&3.38\% \\ 
\flickr        			& 2302K	&  22B    	& 22M		&  33.1M            &  251K    & 0.110  	& 0.108     &837M	& 1.24\%\\  \hline
as-skitter 			& 1.6M	&  16B 		& 11M 		& 53M 		&  160K 	& 0.005	& 0.005 	&28M 	&6.50\%  \\
cit-Patents 		& 3.7M	&  0.3B 		& 16M		& 79M		&  199K  	& 0.067	& 0.066 	&7.51M 	&0.33\%  \\
web-Google  		& 0.8M	&  0.7B 		& 4M 		& 20M 		&  79K  	& 0.055 	& 0.057	& 13.3M	&3.79\% \\
web-NotreDame 	& 0.3M	&  0.3B 		& 1M 		& 5M 		&  42K 	& 0.088 	& 0.088	& 8.91M	&3.93\%  \\
youtube 			& 1.1M	&  1.4B 		& 2M 		& 14M 		&  64K 	& 0.006 	& 0.006 	&3.05M 	&1.86\%  \\
livejournal 		& 5.2M	&  7.5B 		& 48M 		& 205M 		&  473K 	& 0.124 	& 0.118	& 310M	&8.65\%  \\
orkut   			& 3.0M	&  45B 		& 223M 		& 562M  		&  1.2M  	& 0.041	& 0.041 	&627M 	&0.09\% \\ \hline
\end{tabular}
\end{table*}

\textbf{Convergence of estimate:}  \Fig{dblp-final-convergence}\subref{entire-trans} and \Fig{dblp-final-convergence}\subref{entire-trigs} demonstrate convergence of the {\em final estimates} (i.e. for $G_m$) for increasing space. 
We define storage as the  number of edges stored by our algorithm: $|\mbox{\edgeres{}}| + 2\cdot |\mbox{\wedgeres}|$. We first choose $\beta$ in $\set{0.2, 0.4, 0.6, 0.8, 1.0}$ and then vary $\alpha$ in increments of $0.0005$ up to $0.02$. For each setting of $\alpha$ and $\beta$, we plot 5 runs of the algorithm. One can see that both the transitivity and triangles estimates converge rapidly to true values as we increase the space.
\begin{figure*}[tb]
  \centering
  \subfloat[{Transitivity; 1999--2008}]{\includegraphics[width=0.24\textwidth]{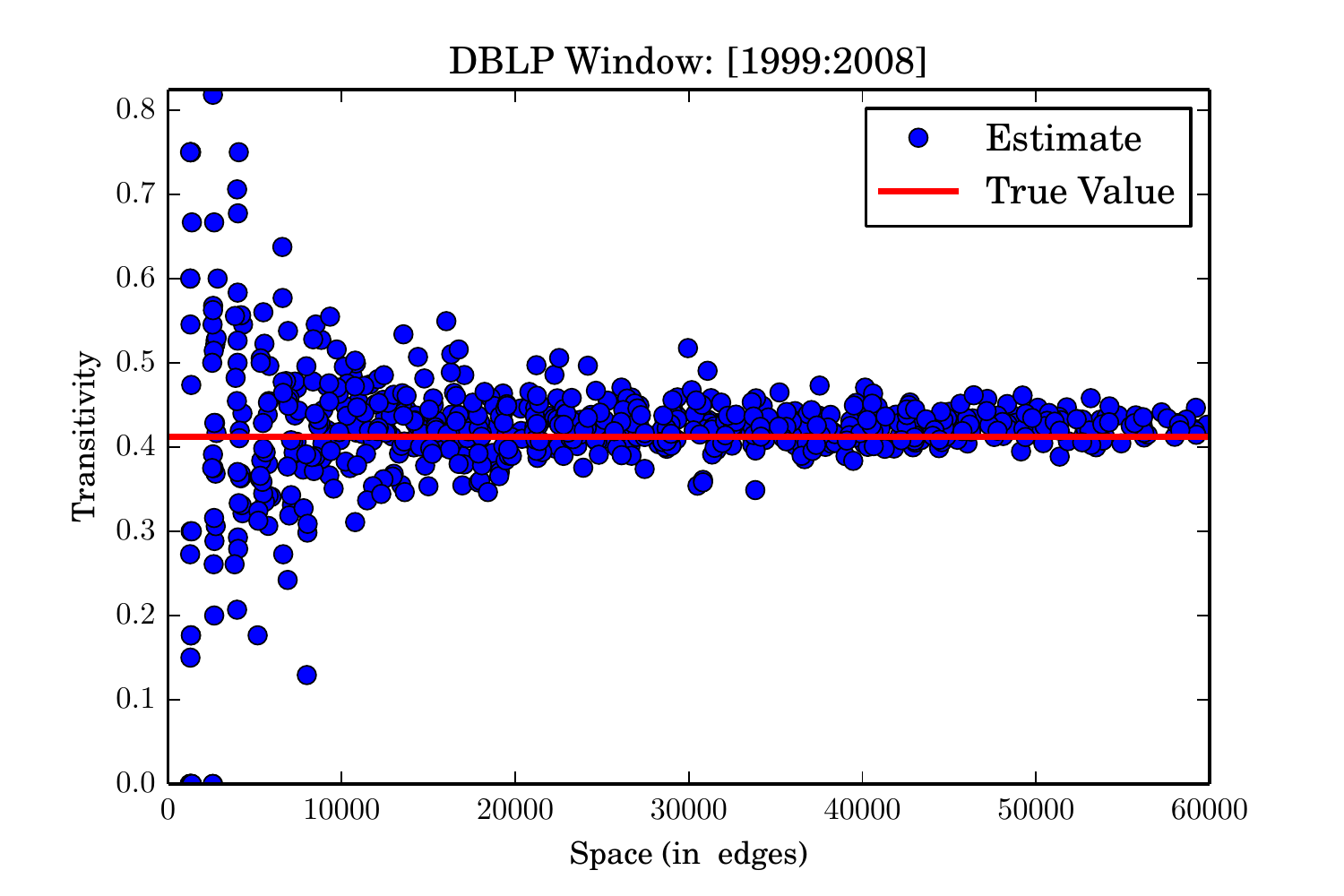}} \ 
    \subfloat[{Transitivity; 1989--2008}]{\includegraphics[width=0.24\textwidth]{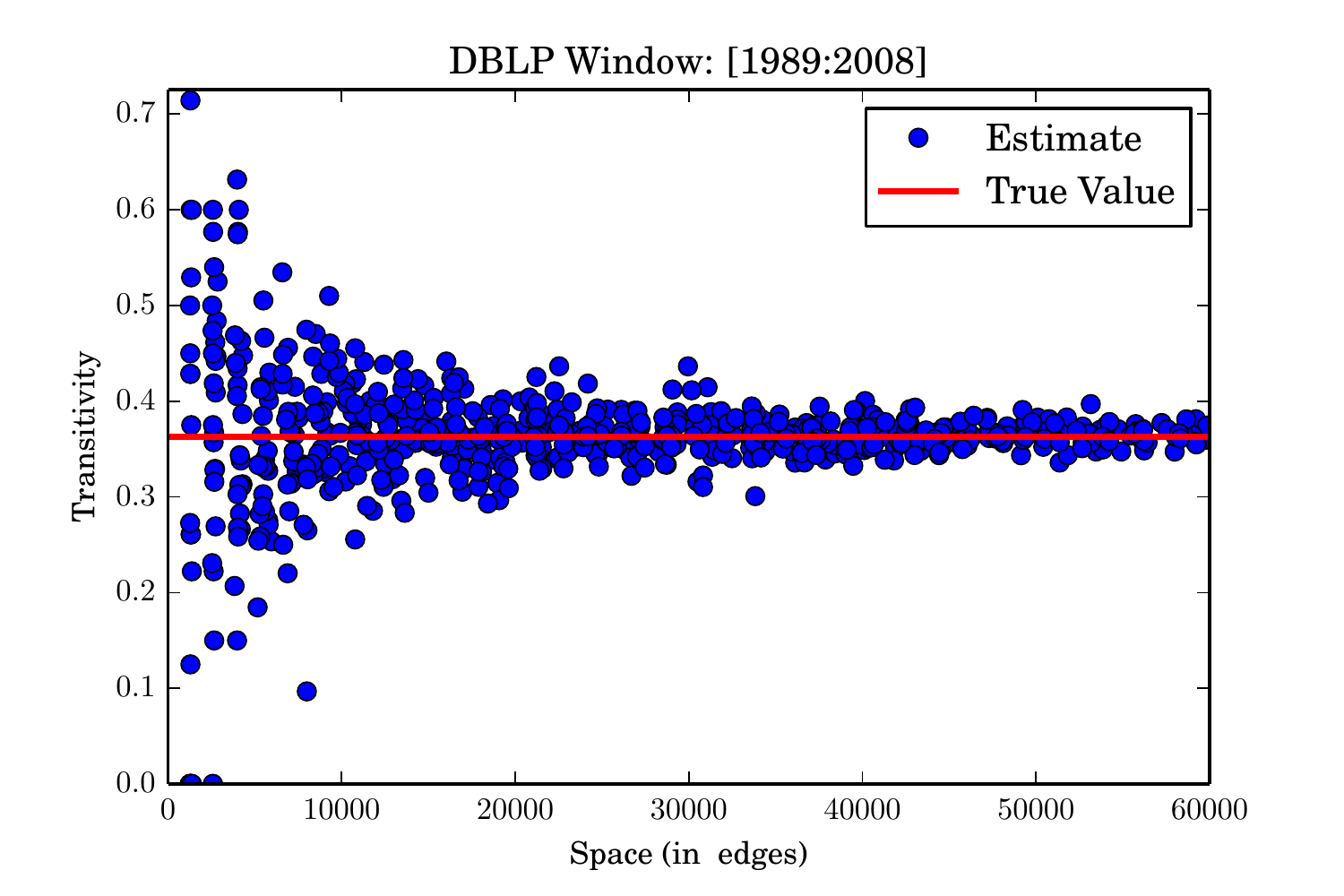}} \ 
  \subfloat[{Transitivity; 1938--2008}]{\includegraphics[width=0.24\textwidth]{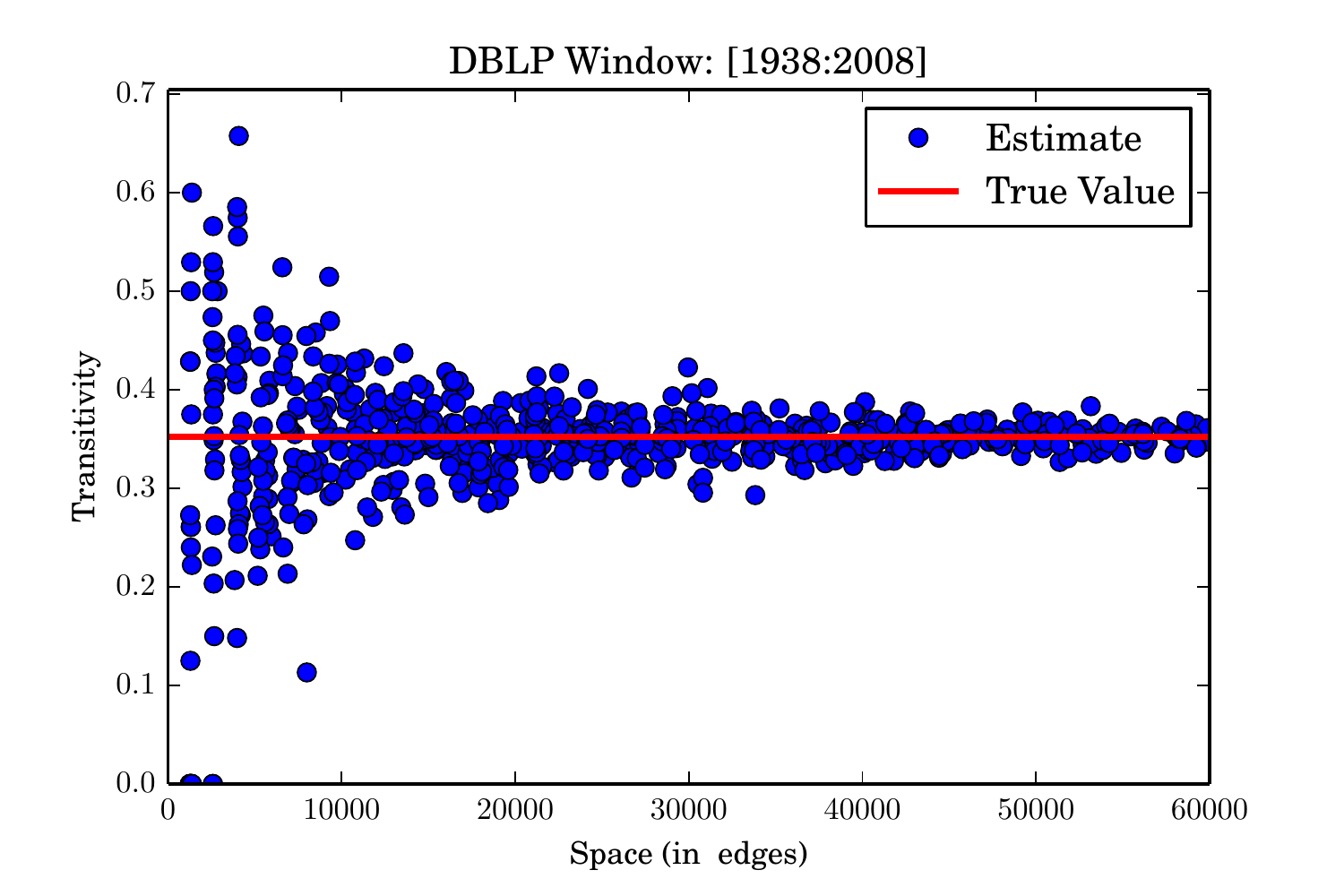}} \ 
        \subfloat[{Transitivity; entire stream}]{\includegraphics[width=0.24\textwidth]{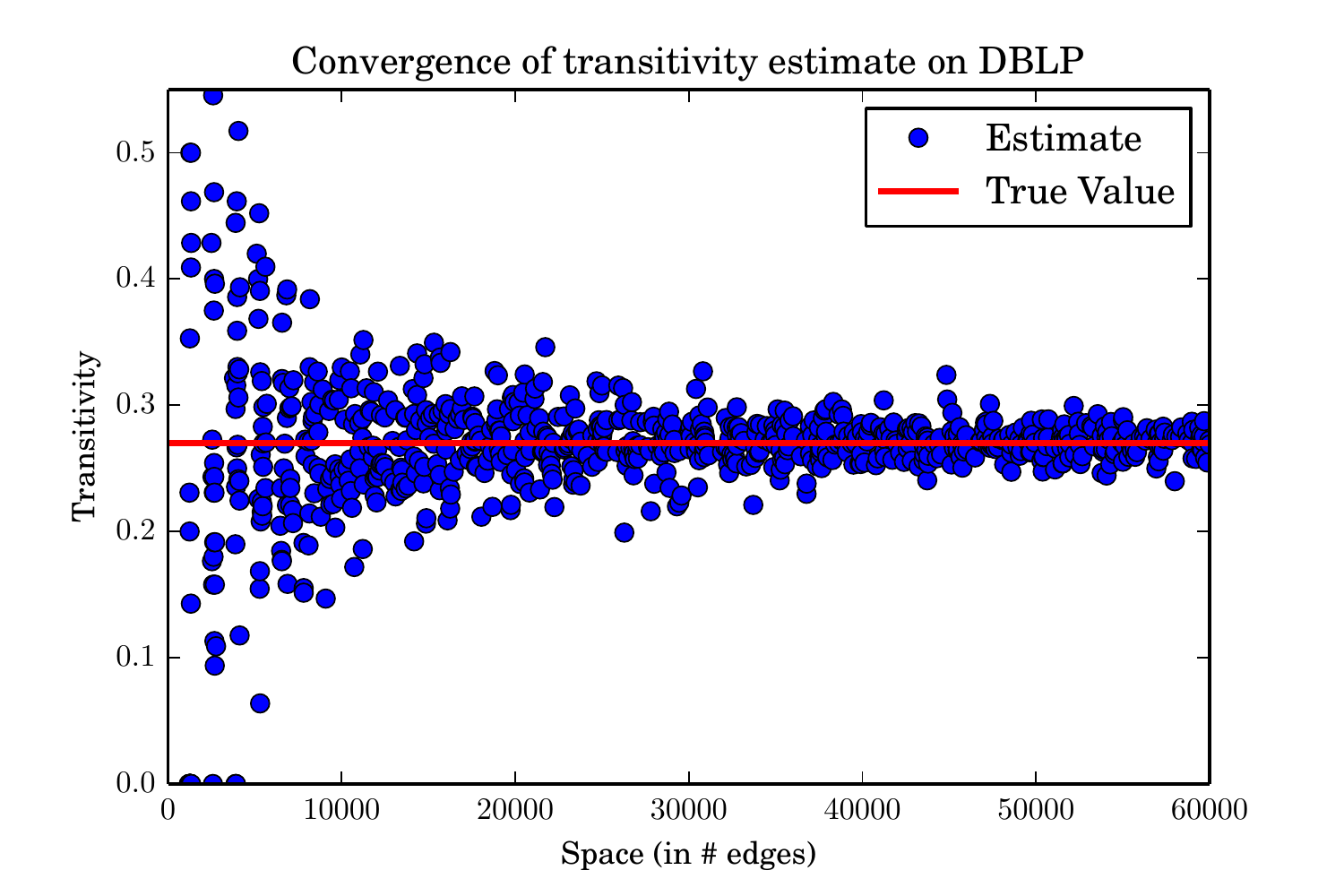}\label{entire-trans}}\
  \subfloat[{Triangles; 1999--2008}]{\includegraphics[width=0.24\textwidth]{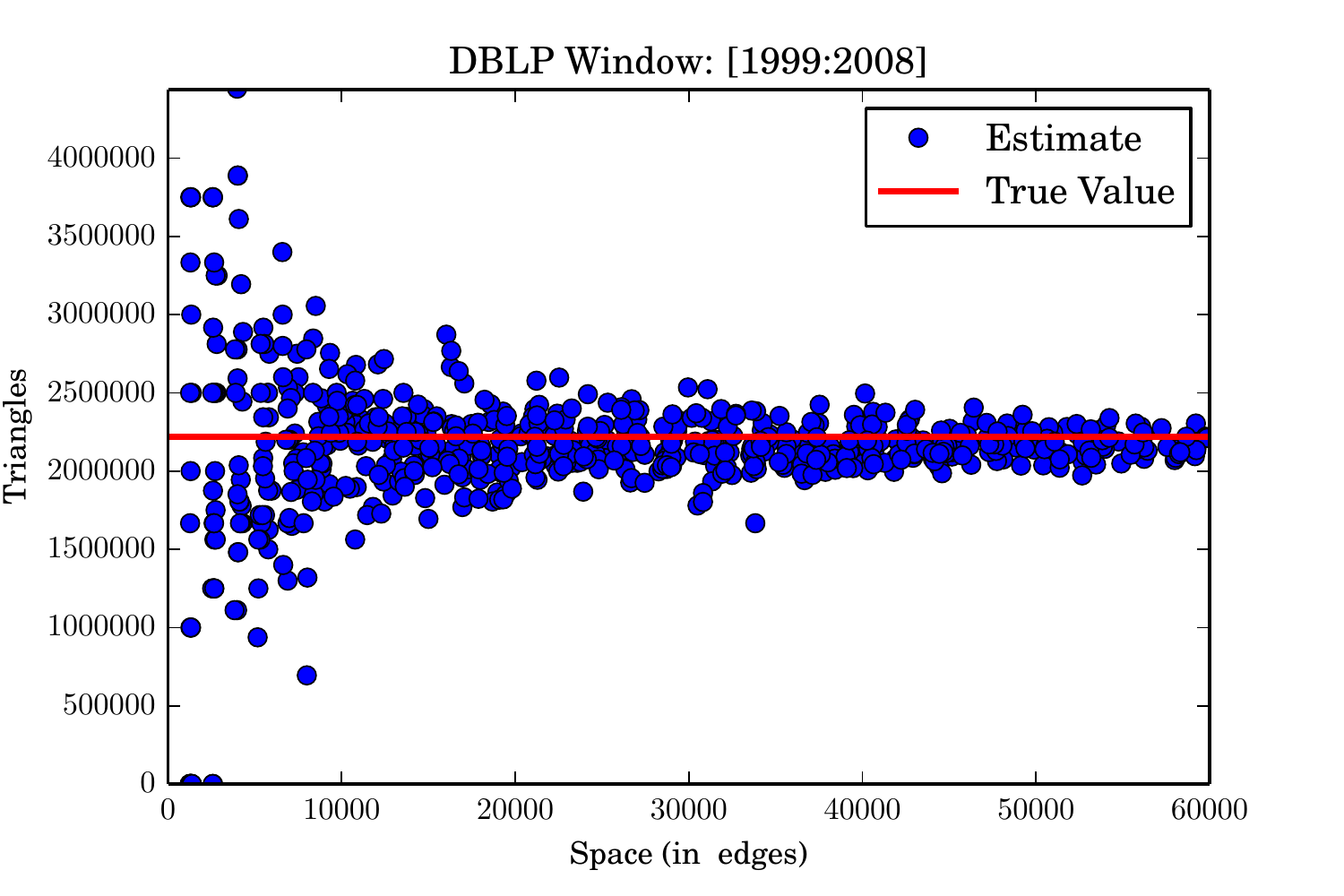}}\
    \subfloat[{Triangles; 1989--2008 }]{\includegraphics[width=0.24\textwidth]{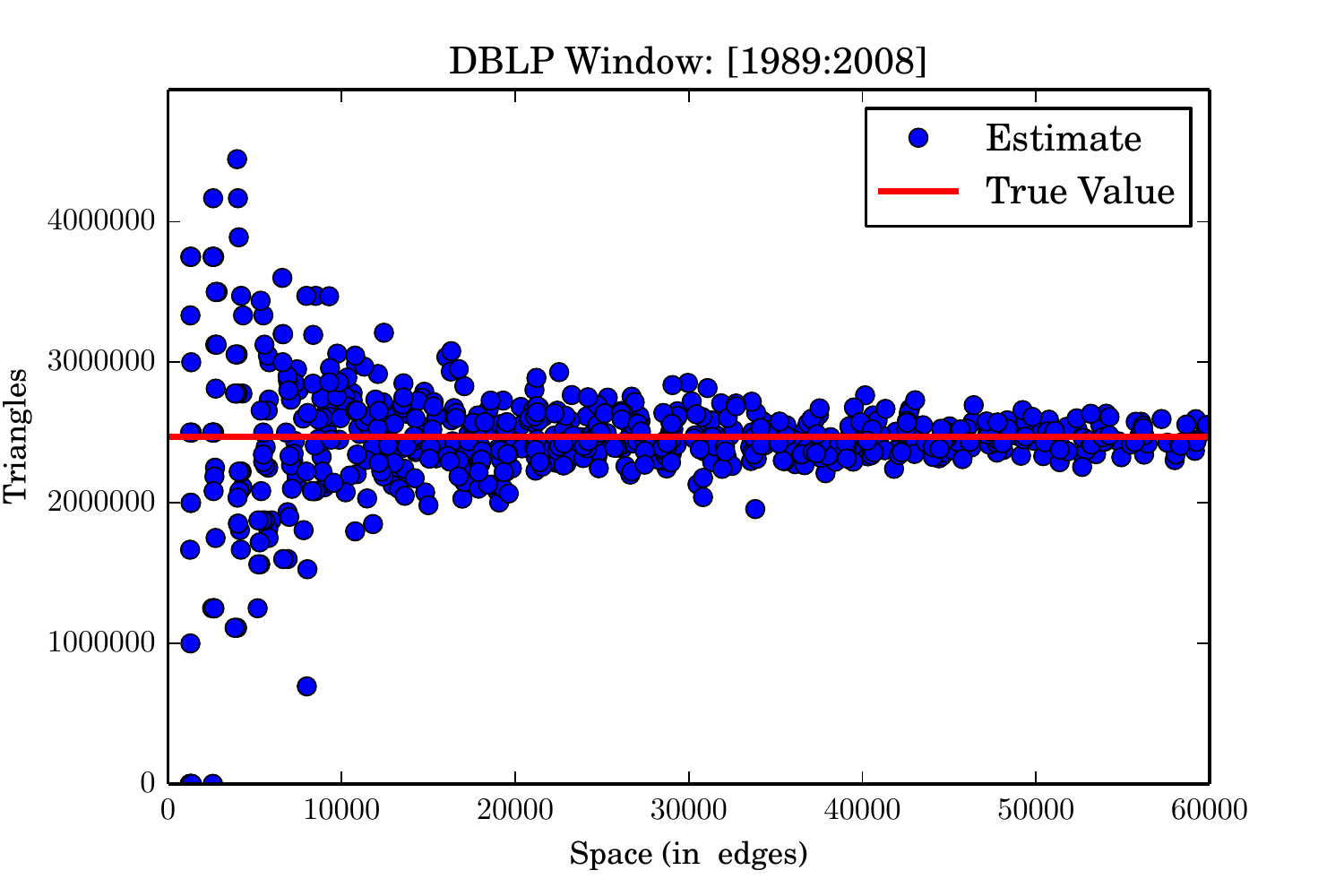}} \ 
     \subfloat[{Triangles; 1938--2008}]{\includegraphics[width=0.24\textwidth]{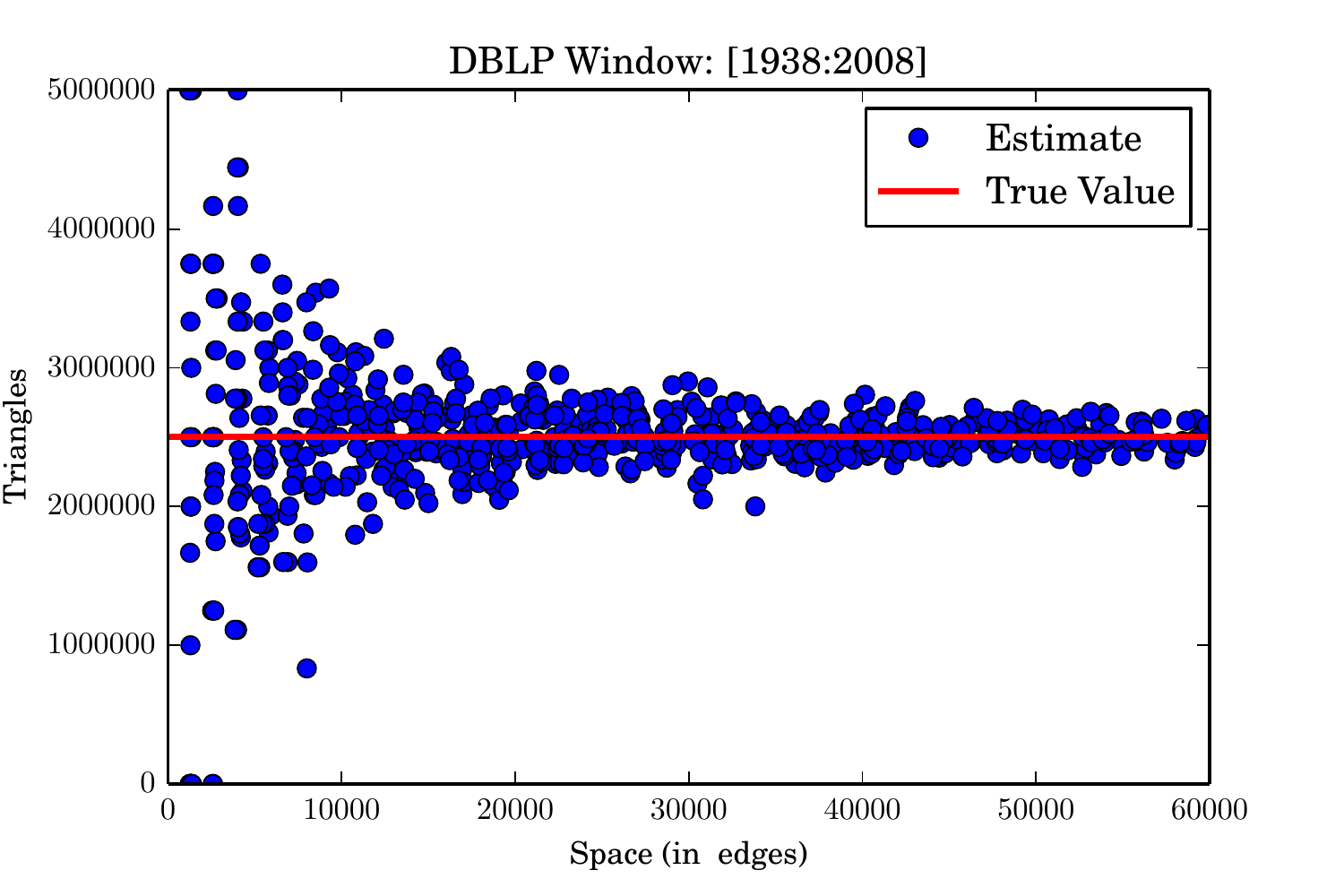}}\
    \subfloat[{Triangles; entire stream}]{\includegraphics[width=0.24\textwidth]{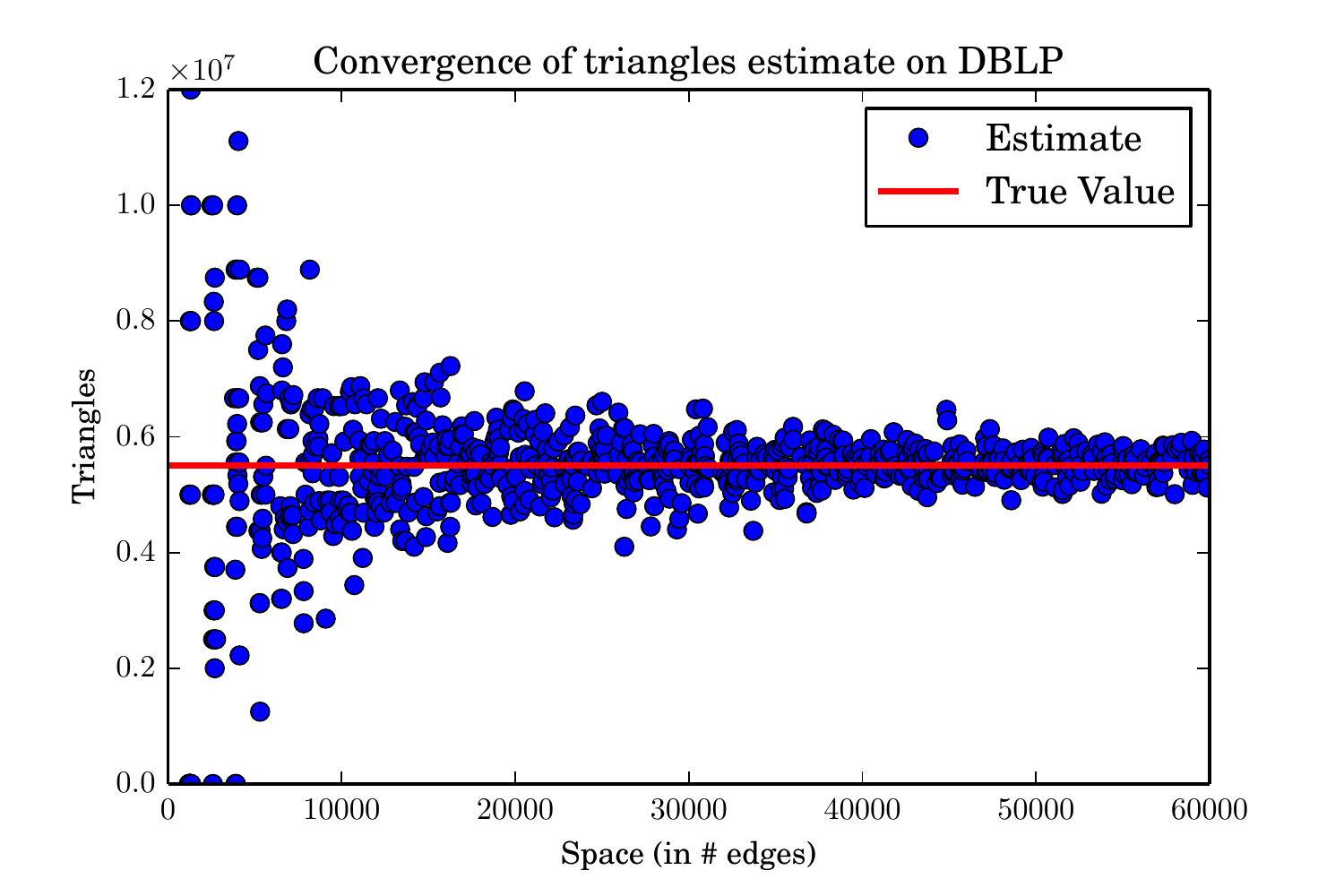}\label{entire-trigs}} \ 
  
   \caption{\DBLP{} convergence: We show that both transitivity and triangles estimates converge to true values as we increase the space. The top row  is for transitivity,  while the bottom row is for triangle counts. The plots are arranged in the order of increasing window length ending at 2008 except for the last column which corresponds to the entire stream length.}
  \label{fig:dblp-final-convergence}
\end{figure*}

Our estimates for various time windows also converge rapidly, as we demonstrate in  \Fig{dblp-final-convergence}.
For these experiments,
we picked specific time windows on \DBLP, namely, 1989--2008, 1999--2008,  and 1938--2008. This is mostly
for demonstrating the convergence of differing window sizes. We chose $\beta$  from $\set{0.2, 0.4, 0.6, 0.8, 1.0}$ and varied $\alpha$ in increments of $0.1\%$ up to $3.0\%$. For each value of $\alpha$ and $\beta$, we give 5 runs of the algorithm. 
In the plots $x$-axis gives increasing space (i.e., increasing $\alpha$) and
the $y$-axis is the estimate.

Across the board, we see rapid convergence as storage increases. For \DBLP, storage of 60K
is enough to guarantee extremely accurate results (relative errors within 5\%), for all the time windows. 
This is even true for the 10 year window, which is quite small compared to the entire stream of data
(\estimate{} will not work for window sizes of a year, since there  are not enough samples from such a window. But the number of edges in a year is small enough to store explicitly).

\textbf{Space usage:}  \Fig{space-usage} shows the space used by our algorithm in terms of parameters $\alpha$ and $\beta$. 
We measure both \edgeres{} and \wedgeres{} for varying values of $\alpha$ and  $\beta$, and plot the predictions
of \Thm{space}. We see almost perfect alignment of the predictions with \Thm{space}. 

\begin{figure*}[t]
  \centering
  \subfloat[{\edgeres{} size}]{\includegraphics[width=0.4\textwidth]{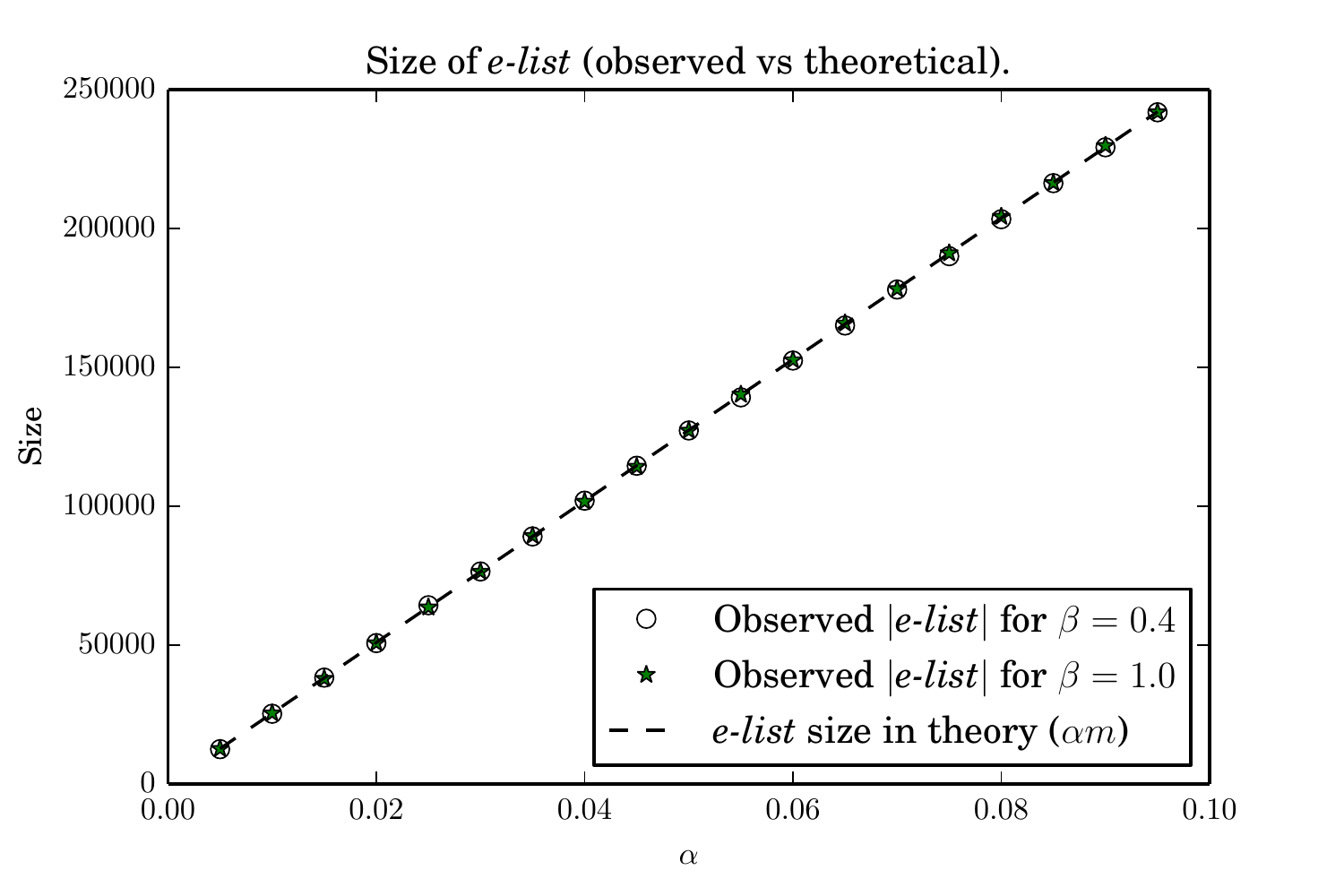}}
  \subfloat[{\wedgeres{}  size}]{\includegraphics[width=0.4\textwidth]{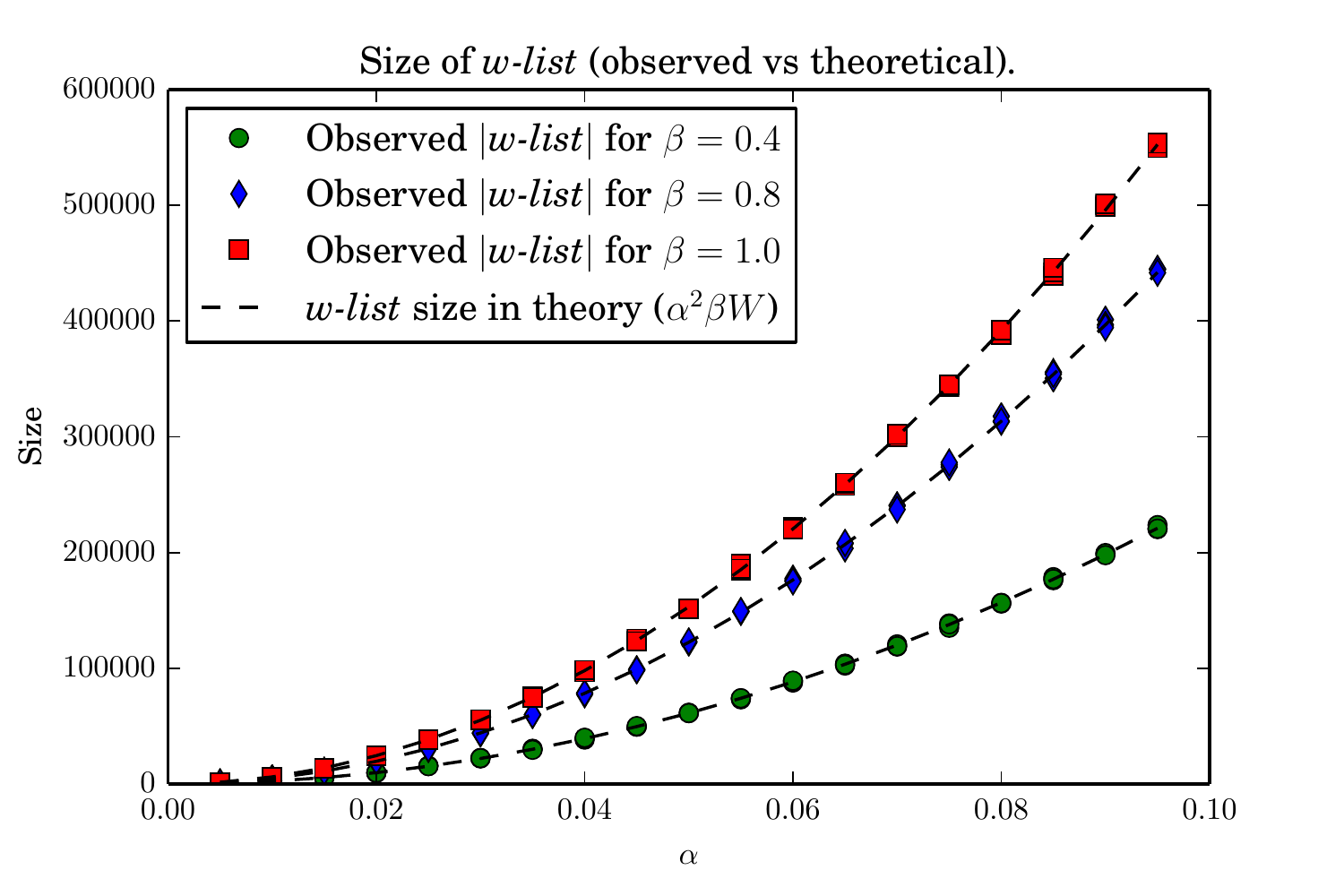}}
   \caption{Number of edges and wedges stored by \edgeres{} and \wedgeres{} respectively. The observed usage matches almost exactly with theoretical predictions.}
  \label{fig:space-usage}
\end{figure*}


\textbf{Comparison with previous work:} We run the algorithms of~\cite{PaTaTi+13},~\cite{JhSePi13}, and~\cite{AhDuNe+14}, using hash based sampling to recreate uniform edge sampling in a multigraph. We first note that our implementations work correctly on the simple graph version of \DBLP, shown
in  \Fig{prev-simple}. All algorithms converge extremely rapidly. When these
algorithms are applied to the multigraph version of \DBLP, then they all converge to incorrect triangle estimates (\Fig{prev-multi}).

\textbf{Tests on a broader data set:} For more validation of \estimate{}, we run it on a large
set of real-world graphs. Most of these graphs are neither temporal nor multigraphs. We construct a multigraph stream
from each graph as follows: every edge $e$ of the graph is independently replicated with probability $1/3$
(specifically $r$ times where $r$ is uniform in $\{2, 4, 8, 16, 32\}$). The stream is obtained by randomly
permuting these multiedges.
For each graph, we only use \estimate{}  record to transitivity and triangle
count of the entire stream (the graph $G[1,m]$). The results are presented in \Tab{run-synthetic}. For these runs, we set $\alpha = 0.01$ and capped the size of wedge reservoir to $50K$ (by choosing $\beta$ appropriately). We observe that transitivity estimates are very sharp (matching the true values up to the third decimal point in many cases). 
The relative error in triangles estimates is less than $3\%$ for most cases and never exceeds $8.7\%$. The overall space used by the algorithm is at most $4\%$ of the number of edges of the underlying {\em simple} graph. We point out that for {\tt orkut} which has nearly half a billion edges (after injecting duplicate edges), the transitivity estimate closely matches with the true value and the relative error in triangles is less than $1\%$. The total storage used is less than $0.5\%$ of the edge stream.

\section{Experiments with time windows} \label{sec:implement}

\estimate{} takes as input a single time window length $\Delta t$. But observe that
the primary data structures \edgeres, \wedgeres, and $X_w$ are independent
of this window. As a result, \estimate{} can handle multiple time windows with the 
\emph{same} data structure. 
We only maintain the latest timestamp for each edge, and do not store any history.
If the time window $[t-\Delta t,t]$ is too small, it is unlikely that \edgeres{}
will have any edges from this window. On the other hand,
small time windows can be stored explicitly to get exact answers. 

%
\textbf{Triangle trends in \DBLP:} In our opinion, the following results are the real achievement of \estimate.
We wish to understand transitivity and triangle trends for \DBLP{}  in various time windows. We focus
on 5-year, 10-year, 15-year, 20-year, and entire history windows. So think of a (say) 5-year sliding time window
in \DBLP, and the aim is to report the transitivity in each such window. Refer to \Fig{dblp-window}
(``All" refers to the window that contains the entire history).
\emph{The algorithm \estimate{} makes a single pass over \DBLP{} without preprocessing and provides results for all these windows at 
every year.} 

The transitivity reveals intriguing trends. Firstly, smaller windows have higher transitivity. It shows that network clustering tends to happen
in shorter time intervals. This is probably because of the affiliation structure of coauthorship networks. 
The increase of triangle counts over time (for the same window size) may not be too surprising, given that the
volume of research increasing. But juxtapose this with 
the \emph{decreasing} of transitivity over time.
This means that (say) the transitivity in 2004--2008 is higher than 2009--2013, even though there
are more papers (and more triangles) in the latter interval. Why is this the case? Is it because of increasing
of interdisciplinary work, which might create more open wedges? Or is it simply some issue with the recording
of \DBLP{} data? Will the decreasing transitivity converge in the future, or do we expect it to simply go to zero?
Can we give a reasonable model of this behavior? We believe that the output of \estimate{}
will lead to many data science questions, and this is the real significance of the algorithm.

\textbf{Triangle trends in  \enron:} In \Fig{enron-w-tri} and \Fig{enron-w-gcc}, we present triangles and transitivity estimates for \enron{} for various windows. For this dataset, we think of a window as being defined by a specified number of past edges. In particular, apart from considering the entire past, we look at windows formed by past 200K, 400K, and 800K edges. Observe that in the beginning of the stream all these windows coincide, since the windows are equivalent. Focusing on the triangles estimate, it is clear that the estimate corresponding to the larger window size dominates that of a smaller window size. What is interesting for \enron{} dataset is that the same ordering is observed even for transitivity estimates. That is, in general, a transitivity estimate curve corresponding to the larger size window dominates the one corresponding to the smaller size. We observe a completely opposite behavior with \DBLP{} transitivity curves, see \Fig{dblp-window}.

Another interesting observation is that in case of \enron{}, the curves for triangles estimates for smaller window lengths flattens out  whereas that in \DBLP{} the curves for triangle estimates continue to rise even for smaller time windows. This indicates that the growth of {\em total number of triangles} is superlinear in \DBLP{}  (with respect to the number of years) whereas it is nearly linear (with respect to the number of edges seen so far) in case of \enron. Indeed the final estimate for the number of triangles in \enron{} is almost the same as the number of edges in the stream.


\ignore{
\section{Conclusions} 

We have described a streaming algorithm  to compute the number of triangles  and the  transitivity of a  multigraph. Our algorithm can seamlessly compute estimates for stream windows of various sizes in real-time. It would be very interesting to carefully study the behavior of \estimate{} on more real-world edge streams, to understand the evolution
of triangles over time. It appears that \estimate{} reveals phenomena at different timescales, which might be an aid to finding the ``right"
window for aggregation. Designing good models for temporal graphs is a big open problem, and our findings on \DBLP{} and \enron{} might provide some
useful information towards that objective.
}

\section*{Acknowledgements} 
We thank Ashish Goel for suggesting the use of hash-function based reservoir sampling. This was a key step towards the development of the final algorithm.

\vspace{-3ex} 
{
\scriptsize
\bibliographystyle{abbrv}


\begin{thebibliography}{10}

\vspace{-3ex}
\bibitem{NeAhKo13}
N.~Ahmed, J.~Neville, and R.~Kompelle.
\newblock Network sampling: From static to streaming graphs.
\newblock {\em ACM TKDD}, 8(2):7:1--7:56, 2013.

\bibitem{AhDuNe+14}
N.~K. Ahmed, N.~Duffield, J.~Neville, and R.~Kompella.
\newblock Graph sample and hold: A framework for big graph analytics.
\newblock In {\em KDD}, 2014.

\bibitem{AhGuMc12}
K.~J. Ahn, S.~Guha, and A.~McGregor.
\newblock Graph sketches: sparsification, spanners, and subgraphs.
\newblock In {\em PODS}, pages 5--14, 2012.

\bibitem{BeBoCaGi08}
L.~Becchetti, P.~Boldi, C.~Castillo, and A.~Gionis.
\newblock Efficient semi-streaming algorithms for local triangle counting in
  massive graphs.
\newblock In {\em KDD}, pages 16--24, 2008.

\bibitem{BuFrLeMaSo06}
L.~S. Buriol, G.~Frahling, S.~Leonardi, A.~Marchetti-Spaccamela, and C.~Sohler.
\newblock Counting triangles in data streams.
\newblock In {\em PODS}, pages 253--262, 2006.

\bibitem{Burt04}
R.~S. Burt.
\newblock Structural holes and good ideas.
\newblock {\em American J. Sociology}, 110(2):349--399, 2004.

\bibitem{Co88}
J.~S. Coleman.
\newblock Social capital in the creation of human capital.
\newblock {\em American J. Sociology}, 94:S95--S120, 1988.

\bibitem{CormodeM05}
G.~Cormode and S.~Muthukrishnan.
\newblock Space efficient mining of multigraph streams.
\newblock In C.~Li, editor, {\em PODS}, pages 271--282. ACM, 2005.

\bibitem{DuPiKo12}
N.~Durak, A.~Pinar, T.~G. Kolda, and C.~Seshadhri.
\newblock Degree relations of triangles in real-world networks and graph
  models.
\newblock In {\em CIKM'12}, 2012.

\bibitem{EcMo02}
J.-P. Eckmann and E.~Moses.
\newblock Curvature of co-links uncovers hidden thematic layers in the {World
  Wide Web}.
\newblock {\em PNAS}, 99(9):5825--5829, 2002.

\bibitem{GS12}
D.~F. Gleich and C.~Seshadhri.
\newblock Vertex neighborhoods, low conductance cuts, and good seeds for local
  community methods.
\newblock In {\em KDD}, 2012.

\bibitem{JhSePi13}
M.~Jha, C.~Seshadhri, and A.~Pinar.
\newblock A space efficient streaming algorithm for triangle counting using the
  birthday paradox.
\newblock In {\em KDD}, KDD '13, pages 589--597, New York, NY, USA, 2013. ACM.

\bibitem{JoGh05}
H.~Jowhari and M.~Ghodsi.
\newblock New streaming algorithms for counting triangles in graphs.
\newblock In {\em COCOON}, pages 710--716, 2005.

\bibitem{KaMeSaSu12}
D.~M. Kane, K.~Mehlhorn, T.~Sauerwald, and H.~Sun.
\newblock Counting arbitrary subgraphs in data streams.
\newblock In {\em ICALP}, pages 598--609, 2012.

\bibitem{DBLP}
M.~Ley.
\newblock Digital bibliography \& library project.
\newblock \url{http://www.informatik.uni-trier.de/~ley/db/}.

\bibitem{Mac12}
S.~Macskassy.
\newblock Mining dynamic networks: The importance of pre-processing on
  downstream analytics.
\newblock In {\em P. Intl W. Mining Communities and People Recommenders}, 2012.

\bibitem{McG14}
A.~McGregor.
\newblock Graph stream algorithms: a survey.
\newblock {\em ACM SIGMOD Record}, 43:9--20, 2014.

\bibitem{Milo2002}
R.~Milo, S.~{Shen-Orr}, S.~Itzkovitz, N.~Kashtan, D.~Chklovskii, and U.~Alon.
\newblock Network motifs: Simple building blocks of complex networks.
\newblock {\em Science}, 298(5594):824--827, 2002.

\bibitem{mislove-2008-flickr}
A.~Mislove, H.~S. Koppula, K.~P. Gummadi, P.~Druschel, and B.~Bhattacharjee.
\newblock Growth of the flickr social network.
\newblock In {\em ACM SIGCOMM W. Social Networks}, August 2008.

\bibitem{PaTaTi+13}
A.~Pavan, K.~Tangwongsan, S.~Tirthapura, and K.-L. Wu.
\newblock Counting and sampling triangles from a graph stream.
\newblock In {\em VLDB}, 2013.

\bibitem{Po98}
A.~Portes.
\newblock Social capital: Its origins and applications in modern sociology.
\newblock {\em Annual Rev. Sociology}, 24(1):1--24, 1998.

\bibitem{graphrepository2013}
R.~Rossi and N.~Ahmed.
\newblock Network repository, 2013.

\bibitem{ScWa05-2}
T.~Schank and D.~Wagner.
\newblock Approximating clustering coefficient and transitivity.
\newblock {\em J. Graph Algorithms and Applications}, 9:265--275, 2005.

\bibitem{SeKoPi11}
C.~Seshadhri, T.~G. Kolda, and A.~Pinar.
\newblock Community structure and scale-free collections of {Erd\"os-R\'enyi}
  graphs.
\newblock {\em Physical Review E}, 85(5):056109, May 2012.

\bibitem{SePiKo13}
C.~Seshadhri, A.~Pinar, and T.~G. Kolda.
\newblock Triadic measures on graphs: The power of wedge sampling.
\newblock In {\em SDM}, 2013.

\bibitem{SPK14}
C.~Seshadhri, A.~Pinar, and T.~G. Kolda.
\newblock Wedge sampling for computing clustering coefficients and triangle
  counts on large graphs?
\newblock {\em Statistical Analysis and Data Mining}, 7(4):294--307, 2014.

\bibitem{Snap}
{SNAP}.
\newblock Stanford network analysis project, 2013.
\newblock Available at \url{http://snap.stanford.edu/}.

\bibitem{TaPaTi13}
K.~Tangwongsan, A.~Pavan, and S.~Tirthapura.
\newblock Parallel triangle counting in massive streaming graphs.
\newblock In {\em CIKM}, 2013.

\bibitem{TsKaMiFa09}
C.~E. Tsourakakis, U.~Kang, G.~L. Miller, and C.~Faloutsos.
\newblock Doulion: counting triangles in massive graphs with a coin.
\newblock In {\em KDD}, pages 837--846, 2009.

\bibitem{FoDeCo10}
B.~F. Welles, A.~V. Devender, and N.~Contractor.
\newblock Is a friend a friend?: {Investigating} the structure of friendship
  networks in virtual worlds.
\newblock In {\em CHI-EA'10}, pages 4027--4032, 2010.

\end{thebibliography}
}

\clearpage 
\ifnum\supp=1
\end{document}
\else
\begin{appendix}
\ignore{
\twocolumn[
\begin{center}
 {\Large 
 Supplementary Material for  the paper \\[2ex]
 Counting Triangles in Real-World Graph Streams: Dealing with Repeated Edges and 
Time Windows \\[3ex]

 Madhav Jha, C. Seshadhri, and Ali Pinar
\medskip
}
\end{center}
]
}
\section{Proof of Theorems in Section 3.2}

For proofs in this section we recall the following lemma. 
\problemma*
\noindent Next we restate and prove \Thm{space}.
\spacethm*

\begin{proof} 
For each edge $e$, let $Z_e$ be the indicator for $e$ being in \edgeres{} at time $t$.
The expected size of \edgeres{} is $\E[\sum_{e \in E(G[1,t])} Z_e]$. By \Lem{prob}, $\E[Z_e] = \alpha$, and linearity of expectation completes the proof.
An identical argument holds for \wedgeres.
\qed
\end{proof}

\wedgethm*

\begin{proof} For any wedge $w \in W_t$, let $Y_w = 1$
if $w \in$ \wedgeres{} and $0$ otherwise. Note that $\widehat{W}_t = (\alpha^2\beta)^{-1} \sum_{w \in W_t} Y_w$.
We have $\E[Y_w] = \alpha^2\beta$ by \Lem{prob}. By linearity of expectation, 
\begin{eqnarray*}
\E[\widehat{W}_t] & = & (\alpha^2\beta)^{-1}\E[\sum_{w \in W_t} Y_w] \\
& = & (\alpha^2\beta)^{-1}\sum_{w \in W_t} \E[ Y_w] = |W_t| 
\end{eqnarray*}
\qed
\end{proof}

Finally, we prove the concentration theorem.
\concthm*

The most important step is to prove a variance bound for $\widehat{W}_t$
and $\widehat{T}_t$. After this, the proofs follow from a routine application
of Chebyschev's inequality.

\begin{restatable}{lemm}{varbound}
\label{lem:var} $\max(Var[\widehat{W}_t],Var[\widehat{T}_t])$ $\leq (\alpha^2\beta)^{-1}|W(G_t)|$ $+ 8\alpha^{-1}|W(G_t)|^{3/2}$.
\end{restatable}

\begin{proof} We deal with $\widehat{W}_t$ first.
\begin{eqnarray*}
	Var[\widehat{W}_t] & = & \E[(\widehat{W}_t)^2] - (\E[\widehat{W}_t])^2 \\
	& = & (\alpha^2\beta)^{-2}\E[\sum_{w \in W_t} \sum_{x \in W_t} Y_w Y_x] - |W_t|^2
\end{eqnarray*}
The double summation can be split based on three cases: (i) $w = x$, (ii) $w$ and $x$ are disjoint (they do not share an edge), 
and (iii) $w$ and $x$ have a common edge. For convenience, we will use $\sum_w$
as shorthard for $\sum_{w \in W_t}$. We use the definition of indicator $Y_w$
from \Thm{wedge}.
\begin{eqnarray*}
 & & \E[\sum_{w} \sum_{x} Y_w Y_x] \\
 & = & \sum_{w} \E[Y^2_w] + \sum_{w \cap x = \emptyset} \E[Y_wY_x] + \sum_{w \cap x \neq \emptyset}\E[Y_wY_x] 
\end{eqnarray*}
The first and second are relatively easy to deal with. Since $Y_w$ is an indicator, $Y^2_w = Y_w$ and $\sum_{w} \E[Y_w] = (\alpha^2\beta) |W_t|$.
When $w \cap x = \emptyset$, note that $Y_w$ and $Y_x$ are independent. This is because we assume
that $hash$ is a random function. Hence,
\begin{eqnarray*}
\sum_{w \cap x = \emptyset} \E[Y_wY_x] & = & \sum_{w \cap x = \emptyset} \E[Y_w]\E[Y_x] \leq \sum_{w,x} \E[Y_w]\E[Y_x] \\
& = & (\sum_w \E[Y_w])^2 = (\alpha^2\beta)^2 |W_t|^2
\end{eqnarray*}
Now for the interesting part. Suppose $w \cap x \neq \emptyset$, so $w = \{e_1,e_2\}$ and $x = \{e_1,e_3\}$.
The product $Y_wY_x$ is $1$ iff $e_1, e_2, e_3$ are all in \edgeres{} and both $w$ and $x$ get selected
in \wedgeres. The probability of this is $\alpha^3\beta^2$. How many pairs of wedges $w \cap x \neq \emptyset$
are there? This is exactly $\sum_i {d_i\choose 3}$, where $d_i$ is the degree of vertex $i$ in $G_t$. 
In the following, we use the fact that the $\ell_3$-norm is smaller than the $\ell_2$-norm.
(We also use the bound $\sum_i d_i^2 \leq 4|W_t|$, which follows because $\sum_i  d_i^2 = 2 (|W_t| + |E_t|)$ and by the statement of the theorem $|W_t| \geq |E_t|$.)
\begin{eqnarray*}
\sum_{w \cap x \neq \emptyset} \E[Y_wY_x] & = & (\alpha^3\beta^2) \sum_i {d_i\choose 3} \\
& \leq & (\alpha^3\beta^2) \sum_i d^3_i \\
& \leq & (\alpha^3\beta^2) (\sum_i d^2_i)^{3/2} \\
& \leq & 8(\alpha^3\beta^2) |W_t|^{3/2}
\end{eqnarray*}
Putting it all together,
\begin{eqnarray*}
	Var[\widehat{W}_t] 	& \leq & (\alpha^2\beta)^{-2}[(\alpha^2\beta)|W_t| + (\alpha^2\beta)^2|W_t|^2\\
	& & + 8(\alpha^3\beta^2) |W_t|^{3/2}] - |W_t|^2 \\
	& = & (\alpha^2\beta)^{-1}|W_t| + 8\alpha^{-1}|W_t|^{3/2}
\end{eqnarray*}
Note that $\widehat{T}_t = \sum_w X_w$. We apply an argument identical to 
that above for $Var[\widehat{T}_t]$.
\end{proof}

\Thm{conc} follows fairly directly from the variance bound.

\begin{proof} (of \Thm{conc})
To prove a concentration bound, we will use Chebyschev's inequality. Let $Var[\widehat{W}_t]$ be the variance
of $\widehat{W}_t$. Then $\Pr[|\widehat{W}_t - \E[\widehat{W}_t]| > h] \leq Var[\widehat{W}_t]/h^2$.
Using \Lem{var},
\begin{eqnarray*}
	& & \Pr[|\widehat{W}_t - |W_t|| > \gamma|W_t|] \\
	& \leq & [(\alpha^2\beta)^{-1}|W_t| + 8\alpha^{-1}|W_t|^{3/2}]/(\gamma^2|W_t|^2) \\
	& = & 1/(\alpha^2\beta|W_t|\cdot\gamma^2) + 8/(\alpha|W_t|^{1/2}\cdot\gamma^2)
\end{eqnarray*}
Since $\alpha^2\beta|W_t| \geq 1/\gamma^6$, $\alpha|W_t|^{1/2} \geq 1/\gamma^3$.
Plugging this bound in, the final probability is at most $\gamma$.

An identical argument holds for $\widehat{T}_t$.
\end{proof}

We apply a Bayes' rule argument to prove bounds of $\widehat{\trans}_t$.
\begin{restatable}{thmm}{transthmm}
\label{thm:trans} Assume the conditions of \Thm{conc}.
$|\E[\widehat{\trans}_t] - \trans_t| \leq 10\gamma$ and
$\Pr[|\widehat{\trans}_t - \trans_t| > 8\gamma] < 4\gamma$.
\end{restatable}

 \begin{proof} (of \Thm{trans}) We have $\widehat{\trans}_t = 3\widehat{T}_t/\widehat{W}_t$
if $\widehat{W}_t \neq 0$ and $0$ otherwise. Let $\cE$ denote the event
that $|\widehat{W}_t - |W_t|| \leq \gamma|W_t|$ and $|\widehat{T}_t - |T_t|| \leq \gamma|W_t|$.

Conditioned on $\cE$, 
$$ 3\widehat{T}_t/\widehat{W}_t \leq (3|T_t| + 3\gamma|W_t|)/(1-\gamma)|W_t|
\leq (1+2\gamma)\trans_t + 6\gamma $$
Similarly, conditioned on $\cE$ $3\widehat{T}_t/\widehat{W}_t \geq (1-2\gamma)\trans_t - 6\gamma$.
By \Thm{conc}, $\Pr[\cE] \geq 1 - 2\gamma$. This proves that
$\Pr[\widehat{\trans}_t] - \trans_t| > 8\gamma] \leq \Pr[\overline{\cE}] \leq 2\gamma$.
To bound the expectation, we simply use Bayes' rule.
\begin{eqnarray*} 
\E[\widehat{\trans}_t] = \E[\widehat{\trans}_t | \cE] \Pr[\cE] + \E[\widehat{\trans}_t | \overline{\cE}] \Pr[\overline{\cE}]
\end{eqnarray*}
Since $\widehat{\trans}_t \in (0,1)$, the latter term is in the range $(0,2\gamma)$.
This completes the proof.
\end{proof}

\section{Additional Experimental Results} 
\begin{figure*}[thb]
  \centering
  \subfloat[{Transitivity}]{\includegraphics[width=0.35\textwidth]{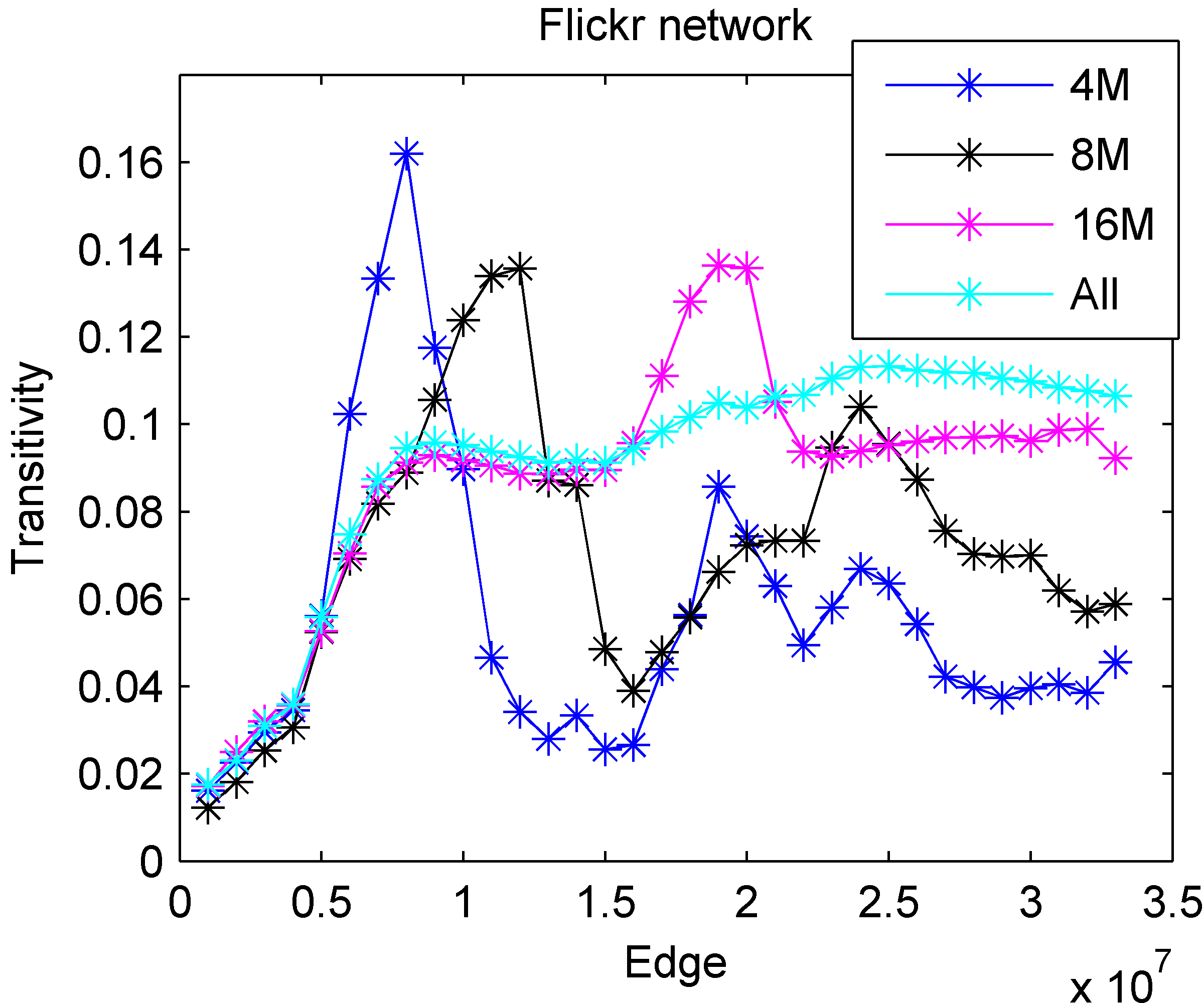}} \ 
  \subfloat[{Triangles}]{\includegraphics[width=0.35\textwidth]{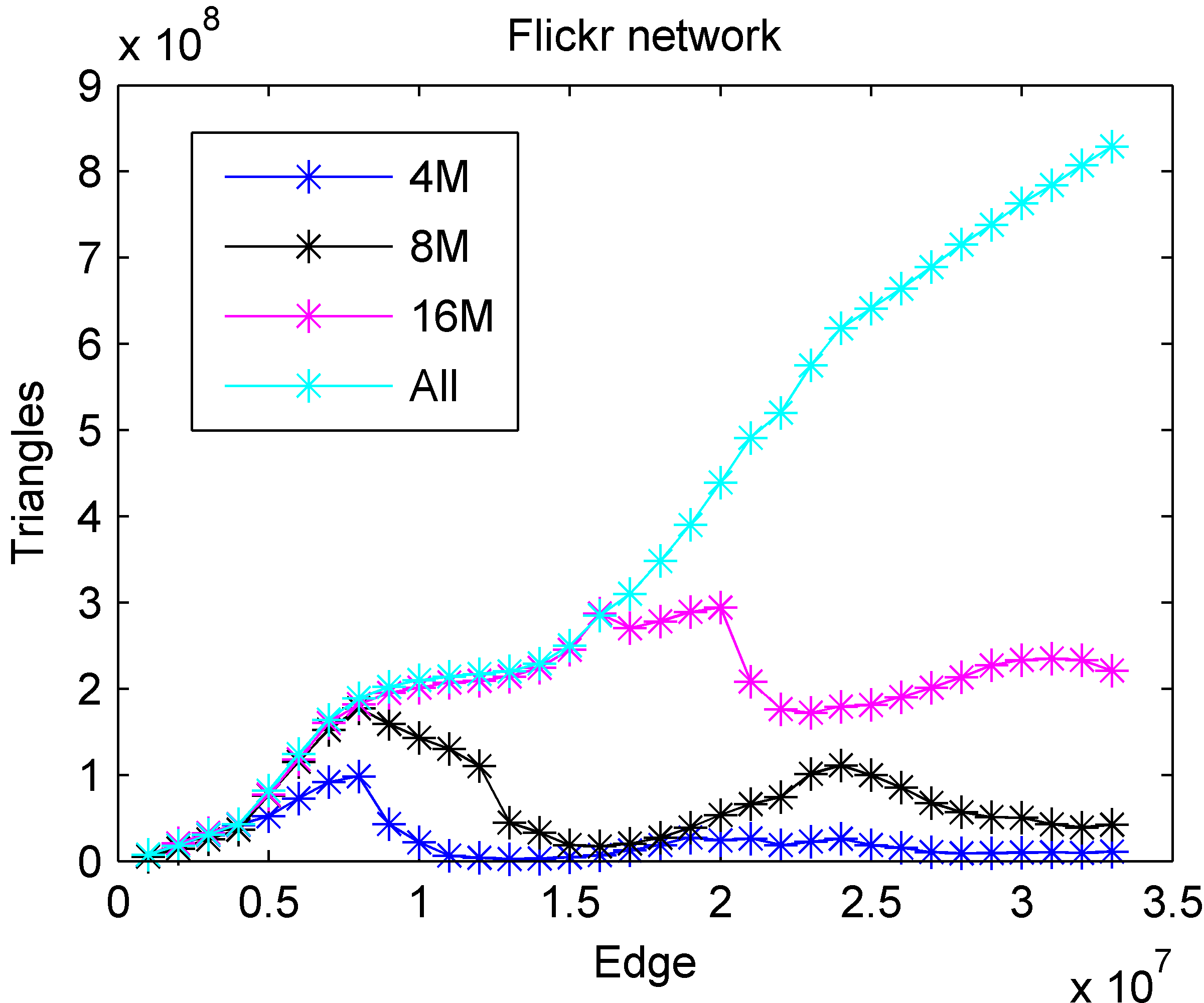}}
   \caption{Similar to the \enron{} experiments, we use  number of edges to define time windows. The algorithm is run with $\alpha = \beta = 0.02$. }
  \label{fig:flickr-window}
\end{figure*}

\textbf{Triangle trends in \flickr:} This is much larger dataset with 33M multiedges. We focus on time windows formed by the past 4M, 8M, 16M, and all history.
These results are given in \Fig{flickr-window}. We are able to get these results with merely 640K edge storage, less than 2\% of the edge stream.

\end{appendix}

\end{document}